\documentclass[acmsmall,screen,nonacm]{acmart}

\setcopyright{none}
\makeatletter
\@ifundefined{hyxmp@parse@acmart}{}{%
  \let\hyxmp@parse@acmart\relax
}
\makeatother

\usepackage{amsmath}
\usepackage{mathtools}
\usepackage{booktabs}
\usepackage{tabularx}
\usepackage{array}
\usepackage{multirow}
\usepackage{enumitem}
\usepackage{graphicx}
\usepackage{xcolor}
\usepackage{tikz}
\usepackage{url}
\usepackage{microtype}
\usepackage{algorithm}
\usepackage{algpseudocode}
\usepackage{placeins}
\usetikzlibrary{arrows.meta,positioning,fit,calc,shapes.geometric}

\hypersetup{
  hidelinks,
  pdftitle={Versioned Transitive Dependency-Closure Binding and Operation-Time Effect Governance for Agent Skills: ClosureBound},
  pdfauthor={Genliang Zhu and Chu Wang},
  pdfsubject={Version, dependency closure, and operation-time effect governance for Agent Skills},
  pdfcreator={LaTeX}
}

\setlist[itemize]{leftmargin=1.25em}
\setlist[enumerate]{leftmargin=1.45em}
\newcolumntype{L}[1]{>{\raggedright\arraybackslash}p{#1}}
\newcolumntype{Y}{>{\raggedright\arraybackslash}X}

\definecolor{opBlue}{HTML}{3B82F6}
\definecolor{opBlueFill}{HTML}{EFF6FF}
\definecolor{opBlueText}{HTML}{1E40AF}
\definecolor{opGreen}{HTML}{22C55E}
\definecolor{opGreenFill}{HTML}{F0FDF4}
\definecolor{opGreenText}{HTML}{15803D}
\definecolor{opAmber}{HTML}{FCD34D}
\definecolor{opAmberFill}{HTML}{FEF3C7}
\definecolor{opAmberText}{HTML}{92400E}
\definecolor{opRed}{HTML}{EF4444}
\definecolor{opRedFill}{HTML}{FEF2F2}
\definecolor{opRedText}{HTML}{DC2626}
\definecolor{opPurple}{HTML}{8B5CF6}
\definecolor{opPurpleFill}{HTML}{F5F3FF}
\definecolor{opPurpleText}{HTML}{7C3AED}

\newtheorem{definition}{Definition}
\newtheorem{lemma}{Lemma}
\newtheorem{theorem}{Theorem}
\newtheorem{proposition}{Proposition}

\newtheorem{assumption}{Assumption}

\newcommand{\ClosureBound}{\textsc{ClosureBound}}

\newcommand{\Norm}{\ensuremath{\mathcal N}}
\newcommand{\Graph}{\ensuremath{\mathcal G}}
\newcommand{\Nodes}{\ensuremath{\mathcal V}}
\newcommand{\Edges}{\ensuremath{\mathcal E}}
\newcommand{\Kinds}{\ensuremath{\mathcal K}}

\newcommand{\denote}[1]{[\![#1]\!]}
\newcommand{\restrict}{\mathbin{\sqsubseteq}}
\newcommand{\mcode}[1]{\ifmmode\text{\normalfont\ttfamily #1}\else{\normalfont\ttfamily #1}\fi}
\newcommand{\statementheading}[1]{%
  \par\smallskip\noindent\textbf{#1.}\par\nobreak\noindent}

\title[Versioned Closure and Effect Governance for Agent Skills]
{Versioned Transitive Dependency-Closure Binding and Operation-Time Effect
Governance for Agent Skills: \ClosureBound}

\author{Genliang Zhu}
\affiliation{%
  \institution{Accentrust}
  \city{Vancouver}
  \country{Canada}
}
\affiliation{%
  \institution{Georgia Institute of Technology}
  \city{Atlanta}
  \state{Georgia}
  \country{USA}
}
\email{research@accentrust.com}

\author{Chu Wang}
\affiliation{%
  \institution{Accentrust}
  \city{Vancouver}
  \country{Canada}
}
\affiliation{%
  \institution{University of Illinois Urbana-Champaign}
  \city{Urbana}
  \state{Illinois}
  \country{USA}
}

\begin{abstract}
Agent Skills combine instructions with files, packages, tools, models, and
services, so operational identity can exceed a signed directory. Recursive or
lazy dependencies may change while root-level evidence remains valid, and
different surfaces may reach the same durable effect.

We present \ClosureBound, a reference monitor that prevents authorization
transfer across material changes to this heterogeneous closure. Its resolver
commits typed graph nodes and topology. Each grant binds an exact closure root,
effect ceiling, purpose/provenance, validity, and epochs. At durability, it re-resolves
closure and state, normalizes the operation into an external-effect IR, and
admits it only if a joint witness satisfies every bound. Supported equivalent
paths share one ceiling.

Assuming complete mediation and discovery, authenticated freshness, sound
normalization, cryptographic binding, and authoritative linearization, we
establish metadata non-authority, closure determinism, version non-inheritance,
effect non-amplification, bound-value freshness, and path invariance. We do not
establish program equivalence or remote-service honesty. A provider-free
implementation matches 40 frozen lifecycle fixtures; 18 kernel contracts and
six mutants cover binding and downgrade cases. Full-profile exploration reaches
84,608 states and 530,752 transitions without a declared invariant violation;
six weakened profiles yield witnesses. A lexical audit of 549 public Skills
(4,872 unique files) finds that 21 of 526 roots with bundled files name every
non-manifest path verbatim, 67 contain links resolving outside their roots, and
no root declares a frontmatter dependencies field. These observations motivate
conservative closure discovery and define concrete targets for broader runtime,
interoperability, efficacy, and production validation.
\end{abstract}


\begin{document}
\maketitle
\hypersetup{pdfcreator={LaTeX}}

\section{Introduction}
\label{sec:introduction}

Agent Skills are becoming a portable unit of procedural capability. The open
format defines a directory with a required \mcode{SKILL.md} file and optional
scripts, references, assets, and arbitrary additional resources; clients
progressively load metadata, instructions, and resources as a task requires
them~\cite{agentskills2026spec,agentskills2026client}. This simplicity is useful:
one Skill can teach an agent how to reconcile invoices, deploy software, fill a
form, or operate a domain-specific tool without retraining the model. It also
means that a capability is assembled over time from natural-language
instructions, executable artifacts, the host agent, and external services.

The installed directory is not the whole capability. A script may resolve a
package at invocation; a reference may direct the model to a helper Skill; a
tool name may resolve to a newer schema; an MCP endpoint may keep the same URL
while its deployment changes; and a Skill may ask the agent to reproduce an
effect through shell or browser automation when the direct tool is unavailable.
Recent measurement confirms that Skill supply chains span mixed
Skill--package--service graphs and that recursive reuse hides material
inventory~\cite{jia2026islands}. Recent attacks likewise exploit metadata,
natural-language compliance rules, embedded models, and apparently routine
side effects rather than only explicit malicious binaries
~\cite{saha2026hood,liu2026payloadless,tie2026badskill,zhuang2026agenttrap}.

\paragraph{The authorization-transfer bug.}
Suppose a user authorizes Skill \mcode{invoice-pay@1.3} to submit one payment to
vendor \(v\), using helper Skill \mcode{bank-wire@2.1}, package
\mcode{iban@4.0}, and service contract \mcode{bank.example/payments@e17}. The
root directory and its signature remain unchanged, but a floating helper
resolves to \mcode{bank-wire@2.2}; the new helper calls an analytics service
and changes the recipient before submission. Alternatively, every artifact is
unchanged but the agent reaches the same payment through a shell client whose
effect is not classified as ``payment.'' In both cases the old authorization is
being transferred to a different operational capability.

Supply-chain attestations answer important questions about where artifacts
came from and whether they changed~\cite{cappos2010tuf,
torresarias2019intoto,slsa2026}. Skill-specific systems add scanning,
permission manifests, versioned dependency resolution and lockfiles, verified
loading, least privilege, behavioral specifications, and runtime probes
~\cite{liu2026wild,pan2026skillguard,shen2026sigil,
wu2026skillscope,li2026vigil,lan2026runtime,bhardwaj2026skillfortify,
dantanarayana2026sigilcompiler,zhan2026autopolicy}.
AIRGuard normalizes heterogeneous actions for action-time authority control;
Cordon and Atomix stage effects and validate or order them at transaction/commit
boundaries; CommitGuard binds durable effects to fresh authority evidence
~\cite{qin2026airguard,chen2026cordon,mohammadi2026atomix,
santos2026temporary}. These contributions establish adjacent foundations.
ClosureBound addresses the following narrower systems question:

\begin{quote}
\emph{What exact transitive capability was authorized, and is that same
capability---not merely the same root name---producing this same normalized
effect now?}
\end{quote}

\paragraph{Approach.}
\ClosureBound resolves an activated Skill into a typed graph whose nodes cover
instructions, bundled files, recursively reused Skills, packages, tool schemas,
model artifacts, services, and effect contracts. Its commitment includes
canonical node evidence and typed edge topology. A task grant binds exactly one
closure root together with purpose/provenance restrictions and an effect
ceiling. Preparation produces only a short-lived reservation. Immediately
before durability, the monitor re-resolves every material identity and epoch,
normalizes the real surface event, checks the current pre-state, and atomically
consumes the reservation. A change leads to denial or explicit
reauthorization; uncertainty after dispatch remains indeterminate.

\begin{figure}[t]
\centering
\begin{tikzpicture}[
  font=\small,
  >=Latex,
  node distance=5mm and 6mm,
  box/.style={rounded corners=2pt, draw, minimum height=8mm, align=center,
              inner xsep=5pt, inner ysep=3pt},
  flow/.style={->, line width=0.7pt},
  bad/.style={->, line width=0.9pt, opRed, dashed}
]
\node[box,draw=opBlue,fill=opBlueFill,text=opBlueText] (grant)
  {task grant\\root Skill \(s@v\)};
\node[box,draw=opPurple,fill=opPurpleFill,text=opPurpleText,right=of grant] (load)
  {load-time\\root check};
\node[box,draw=opAmber,fill=opAmberFill,text=opAmberText,right=of load] (deps)
  {lazy helper/package/\\service update};
\node[box,draw=opGreen,fill=opGreenFill,text=opGreenText,below=9mm of deps] (effect)
  {durable effect\\through path \(p_2\)};
\node[box,draw=opRed,fill=opRedFill,text=opRedText,left=of effect] (gap)
  {old authority\\silently reused};
\draw[flow,opBlue] (grant) -- (load);
\draw[flow,opPurple] (load) -- (deps);
\draw[bad] (deps) -- (effect);
\draw[bad] (effect) -- (gap);
\coordinate (return) at (grant.south |- gap.west);
\draw[bad] (gap.west) --
  node[midway,below=1pt,align=center,text=opRedText,font=\footnotesize]
  {root identity\\still matches} (return) -- (grant.south);
\end{tikzpicture}
\caption{The authorization-transfer gap. Load-time integrity of a root Skill
does not bind a later transitive resolution or an equivalent effect path.
\ClosureBound rebinds the exact closure and normalized effect at operation
time. Blue denotes trusted authority, purple identity, amber change, green the
effect boundary, and red a rejected transfer.}
\Description{A loop from a task grant through a load-time root check, a lazy
dependency update, a durable effect, and silent reuse of old authority.}
\label{fig:gap}
\end{figure}
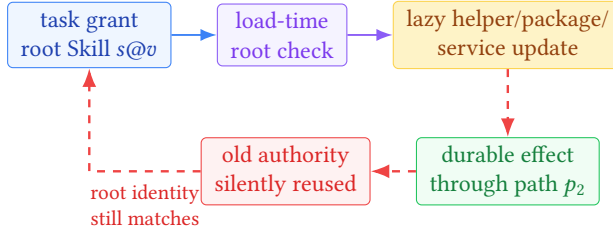

\paragraph{Contributions.}

\begin{enumerate}
  \item We define a root- and endpoint-sensitive commitment for the
  authenticated heterogeneous \emph{operational} closure of an activated Skill,
  spanning instructions, files, packages, tools, models, services, and effect
  contracts. Our contribution is the authorization unit---the exact mixed
  closure used by one mediated effect---rather than a new dependency graph,
  lockfile, or SBOM.

  \item We introduce a joint operation-time witness that binds a trusted task
  grant, exact closure root, current runtime state, purpose/provenance boundary,
  and normalized external effect. The witness establishes authorization
  consistency; Skill benignness remains a separate content-analysis question.

  \item We prove, under explicit assumptions, metadata non-authority, closure
  determinism, version non-inheritance, transitive effect non-amplification,
  operation-time TOCTOU safety, and equivalent-effect path invariance. Every
  theorem states its assumptions and claim boundary.

  \item We provide deterministic algorithms for closure resolution and
  prepare/commit monitoring, together with a typed reauthorization-diff schema
  and fail-closed handling of ambiguity and indeterminate settlement.

  \item We implement and evaluate a provider-free finite reference model on 40
  frozen lifecycle fixtures, five weakened profiles, six executable mutants,
  bounded exploration, and a single-process local cost measurement. This
  evaluation targets mechanism conformance, necessity, finite-state safety, and
  local kernel cost; cross-runtime efficacy and ecosystem coverage require
  separate studies.

  \item We conduct a separately specified static closure audit over 549 public
  Skills from two fixed, independently maintained repository snapshots. It
  characterizes local directory closure and explicit path structure as a fixed
  baseline for runtime-discovery and interoperability studies.
\end{enumerate}

\section{Background and Problem Analysis}
\label{sec:background}

\subsection{A Skill is a progressively resolved capability}

The Agent Skills specification requires only a name and description in
\mcode{SKILL.md}; an experimental \mcode{allowed-tools} field and arbitrary
string-valued metadata are optional. The August 2026 clarification also makes
explicit that the conventional optional directories are non-exhaustive. The
Markdown body has no policy syntax restriction, and scripts are expected to
document or resolve their own dependencies~\cite{agentskills2026spec}. Client
guidance recommends progressive disclosure:
catalog metadata is available first, full instructions are loaded on
activation, and referenced resources are loaded only when required
~\cite{agentskills2026client}. Therefore, the complete operational dependency
set need not exist at discovery or even at activation.

We distinguish:

\begin{description}[leftmargin=1.2em,style=nextline]
  \item[Discovery identity] the untrusted metadata used to retrieve or select a
  Skill.
  \item[Distribution identity] the files, signatures, registry records, and
  supply-chain attestations available before execution.
  \item[Resolved capability identity] the exact material transitive graph that
  can implement a protected effect in one execution.
  \item[Effect identity] the canonical external state transition at the
  durability boundary, independent of the surface used to request it.
\end{description}

Conflating these identities creates two dual errors. A false negative accepts a
changed dependency because the root remains unchanged. A false positive blocks
benign change in discovery order, resolver-local handles, or canonical path
spelling even though the canonical evidence graph is identical. ClosureBound
prioritizes sound authorization: unresolved ambiguity denies, while
canonicalization and representation-equivalent commitments avoid gratuitous
invalidation where the admitted evidence encoding is identical.

\subsection{Why flat manifests and load-time checks are insufficient}

A flat manifest hash can detect a changed listed item, but three structural
problems remain.

\paragraph{Completeness.}
If the manifest omits a nested Skill, lazy package, model file, service
deployment, or effect adapter, hashing the manifest authenticates an incomplete
statement. SPDX explicitly distinguishes asserted relationships from
\mcode{NOASSERTION}; SLSA makes provenance completeness a property of the
producer and notes recursive dependency verification as an additional
step~\cite{spdx2025,slsa2026verify}. ClosureBound carries this distinction into
authorization: unknown material dependencies do not disappear.

\paragraph{Topology.}
An unordered inventory cannot distinguish ``Skill \(A\) invokes package \(P\)
only through sandboxed helper \(B\)'' from ``\(A\) invokes \(P\) directly.''
Edge kinds affect mediation and the effect contract. The commitment must bind
the typed graph, not only component digests.

\paragraph{Time.}
Loading precedes lazy resolution and the real effect. SIGIL correctly binds an
approved artifact to the loaded artifact~\cite{shen2026sigil}; operation-time
authorization correctly rechecks authority at durability
~\cite{santos2026temporary}. ClosureBound composes these insights: the
operation-time witness includes the exact current recursive closure.

\subsection{Why tool-call syntax is not an effect boundary}

One external transition may be requested as a typed API call, generated SDK
code, a shell command, browser interaction, database query, or a helper Skill.
Conversely, the same function name can represent different effects after a
schema or service change. A policy attached only to surface syntax is therefore
neither complete nor path invariant.

ClosureBound does not solve arbitrary program equivalence. Each protected sink
has a trusted adapter that emits a finite effect IR. For a payment, relevant
fields may include payer, payee, currency, amount, funding source,
idempotency scope, and expected ledger state. For a message, they may include
channel, recipients, data classes, and release semantics. Unsupported general
code is denied or confined to a conservative effect ceiling.

\subsection{Research scope}

Table~\ref{tab:scope} separates the central claim from adjacent questions.

\begin{table}[t]
\caption{Research scope and evidence boundary.}
\label{tab:scope}
\small
\begin{tabularx}{\linewidth}{@{}L{0.22\linewidth}YY@{}}
\toprule
Question & In scope & Explicitly out of scope \\
\midrule
Skill identity &
Canonical represented transitive closure and topology &
Complete semantic dependency discovery for arbitrary prose/code \\
Authorization &
Exact closure/effect joint witness under a trusted task grant &
Inferring latent user intent or proving the grant is normatively correct \\
Runtime &
Mediated prepare/commit and current epochs &
Unmediated sinks, compromised monitor, Byzantine infrastructure \\
Effects &
Finite adapter-defined external-effect IR &
General program equivalence, receiver honesty, guaranteed rollback \\
Evidence &
Finite conformance, mutants, bounded model, local kernel cost &
Hosted-model efficacy, ecosystem prevalence, production SLOs \\
\bottomrule
\end{tabularx}
\end{table}

\section{System and Threat Model}
\label{sec:threat}

\subsection{Entities and lifecycle}

The principals are a user or organizational policy authority, an agent runtime,
Skill publishers, package registries, tool/service operators, and external
effect receivers. The lifecycle is:

\[
\mathsf{discover}\rightarrow\mathsf{activate}\rightarrow
\mathsf{resolve}\rightarrow\mathsf{grant}\rightarrow
\mathsf{prepare}\rightarrow\mathsf{commit}\rightarrow
\mathsf{settle}.
\]

Discovery and Skill-authored metadata are untrusted inputs. Distribution
signatures and supply-chain attestations are authenticated evidence, but do not
alone authorize effects. A grant issuer independently combines the current
task certificate with organizational policy. The resolver and monitor form the
enforcement core.

\subsection{Adversary}

The adversary may control Skill metadata, instructions, bundled files, nested
Skills, packages, model artifacts, self-declared permissions, claimed
dependencies, and proposed effects. It may exploit floating versions,
dependency confusion, alternate locators, symlinks, cycles, caches, lazy
loading, service redeployment, schema drift, replay, and prepare/commit races.
It may induce the model to switch to a different mediated path that reaches an
equivalent effect.

The adversary cannot forge authenticated task grants, trusted registry or
service epochs, monitor reservations, or collision-resistant digests. It does
not compromise the canonical resolver, effect normalizers, clock,
linearization store, or policy authority. Direct effects that bypass all
configured sinks are outside the theorem rather than counted as prevented.

\subsection{Assumptions and trusted computing base}

\begin{assumption}[Complete represented dependency closure]
\label{ass:closure}
Every material object that can influence or implement a protected mediated
effect is represented as a node in the closure, or conservatively enclosed by
a stronger node such as a pinned container or attested service deployment.
\end{assumption}

\begin{assumption}[Complete effect mediation]
\label{ass:mediation}
Every protected sink crosses a trusted adapter whose normalizer is sound for
the supported surface. Unsupported or ambiguous effects fail closed.
\end{assumption}

\begin{assumption}[Authoritative freshness]
\label{ass:fresh}
Versions, service/schema/effect-contract epochs, revocation state, time, and
pre-state witnesses come from fail-closed authoritative sources rather than
solely from the Skill.
\end{assumption}

\begin{assumption}[Linearizable commit]
\label{ass:linear}
Revalidation and reservation consumption are atomic at one operation
linearization point, or the adapter consumes an equivalent current fencing
token before the protected effect.
\end{assumption}

\begin{assumption}[Cryptographic binding]
\label{ass:crypto}
The canonical hash is collision and second-preimage resistant, authenticated
statements are unforgeable, and keys/roots are administered correctly.
\end{assumption}

The TCB includes the grant issuer, canonicalization rules, resolvers and
evidence adapters, effect normalizers, version/epoch services, clock,
operation-time monitor, reservation store, and cryptographic primitives.
Remote services are trusted only to the granularity of the evidence they
provide. A stable URL does not prove an unchanged backend.

\subsection{Protected event and outcome states}

The protected event is a monitor-mediated operation reaching its configured
durability boundary. Preparation is not durability. Before dispatch, the
monitor returns a decision in

\[
\mathcal D=\{\mcode{ALLOW\_COMMIT},\mcode{DENY},
\mcode{REAUTHORIZE}\}.
\]

After dispatch, settlement is reported separately as
\(\mathcal O=\{\mcode{COMMITTED},\mcode{ABORTED},
\mcode{INDETERMINATE}\}\). A changed but legitimate closure requires a new
grant; mutating an old grant would destroy non-inheritance. If dispatch may
have occurred but the postcondition is unavailable, the settlement remains
\mcode{INDETERMINATE}.
ClosureBound does not derive actual receiver state from a timeout or an
unauthenticated receipt.

\section{Formal Model}
\label{sec:model}

\subsection{Typed dependency graph}

Let \(s\) be an activated root Skill, \(t\) a time, \(\pi\) an authenticated
evidence profile, \(c_\pi\) its separately versioned canonicalizer, and \(B\) a
trusted pre-resolution effect upper bound compiled without closure-derived
contracts. Its
resolved capability graph is

\[
\Graph_t(s;\pi,c_\pi,B)=(\Nodes_t,\Edges_t,r_s),
\]

where \(r_s\in\Nodes_t\) is the root. Node kinds are

\[
\Kinds=\{\mathsf{skill},\mathsf{instruction},\mathsf{file},
\mathsf{package},\mathsf{tool},\mathsf{model},\mathsf{service},
\mathsf{effectContract}\}.
\]

Edge kinds are
\[
\mathcal R=\{\mathsf{contains},\mathsf{reads},\mathsf{executes},
\mathsf{imports},\mathsf{invokes},\mathsf{loads},\mathsf{resolves},
\mathsf{declaresEffect}\}.
\]

A node evidence record is

\[
q_t(v)=\langle k,\ell,\delta,\chi,\zeta,\epsilon,\sigma\rangle ,
\]

where \(k\) is kind, \(\ell\) a canonical locator (including the relevant
namespace/issuer), \(\delta\) a resolved version or deployment identity,
\(\chi\) a content or contract digest, \(\zeta\) the authenticated digest of
every resolution-context fact that can change outgoing dependencies,
\(\epsilon\) an authoritative epoch, and \(\sigma\) an evidence class and
authentication proof. Resolver-local graph handles are not fields of
\(q_t\). Nodes with the same canonical record, including \(\zeta\), are
interned before commitment. A mutable service with no authenticated deployment
identity is represented only by its declared contract and URL; a policy may
reject that weak evidence class.

\[
q^E_t(u,r,v)=
\langle r,h_V(u),h_V(v),source,\epsilon_E,\sigma_E\rangle
\]

is the authenticated evidence record for an edge: its type and exact endpoint
identities, derivation source, authoritative epoch, and evidence class/proof.
The endpoint hashes are checked after both nodes have been interned; they are
not taken on faith from a manifest.

\begin{definition}[Material closure]
A node \(v\), or typed edge incident to it, is material under
\((\pi,c_\pi,B)\) if it
can alter the selection, arguments, recipient, data, sequencing,
implementation, or settlement interpretation of any \(e\in\denote B\) within
the declared mediation envelope. For a realized effect \(e\),
\(\mathcal I_{\Graph}(e)\) is the union of authenticated root-to-sink influence
paths that produced or can implement that event.
\end{definition}

The resolver can conservatively include more nodes. Omission threatens
soundness; over-approximation causes reauthorization or false denial but does
not expand authority.

\subsection{Canonical graph commitment}

Locators are normalized by a kind-specific total function \(C_k\). Relative
paths are resolved against the authenticated parent root; symlinks are
dereferenced or rejected according to policy; package coordinates include
registry namespace; service URLs use a canonical URI form; and effect
contracts have explicit schema versions.

Cycles are possible through recursive Skills or services. Let
\(\operatorname{SCC}(\Graph)\) be the deterministic condensation DAG.
For each node,

\[
h_V(v)=H(\mcode{CB:node:v1}\parallel
\operatorname{enc}(q_t(v))).
\]

For each typed edge,

\[
h_E(u,r,v)=H(\mcode{CB:edge:v1}\parallel
\operatorname{enc}(q^E_t(u,r,v))).
\]

Within each strongly connected component \(X\), define

\[
h_X=H(\mcode{CB:scc:v1}\parallel
\operatorname{sort}\{h_V(v):v\in X\}\parallel
\operatorname{sort}\{h_E(e):e\in E[X]\}).
\]

For every inter-component edge \(e=(u,r,v)\), with \(u\in X_i\) and
\(v\in X_j\), bind both component and original endpoint identities:

\[
h_C(e)=H(\mcode{CB:component-edge:v1}\parallel h_{X_i}\parallel
h_V(u)\parallel r\parallel h_{X_j}\parallel h_V(v)\parallel h_E(e)).
\]

The closure commitment includes the condensed topology and the exact root node:

\begin{equation}
\begin{aligned}
\rho_t(s;\pi,c_\pi,B)=H(&\text{\texttt{CB:closure:v1}}\parallel h_V(r_s)
  \parallel h_{X_r}\\
  &{}\parallel\operatorname{sort}\{h_X\}
  \parallel\operatorname{sort}\{h_C(e):e\in E,\ X(u)\ne X(v)\}).
\end{aligned}
\label{eq:closure-root}
\end{equation}
Every \(\operatorname{sort}\{\cdot\}\) above denotes a canonical sorted
multiset, not set deduplication; parallel authenticated edge records retain
their multiplicity.

\begin{definition}[Exact resolvability]
\label{def:exact}
\(\Graph_t(s;\pi,c_\pi,B)\) is exactly resolvable under profile \(\pi\) iff every
material node and edge has a profile-admissible authenticated evidence record,
all endpoint checks and root reachability checks succeed, no material version
is floating or \mcode{NOASSERTION}, canonicalization succeeds, all required
epochs are current, and every dynamic load is either enumerated or enclosed by
an admissible stronger identity.
\end{definition}

\begin{lemma}[Closure determinism]
\label{lem:determinism}
For a fixed canonicalization profile, evidence records, typed graph, and root,
Equation~\eqref{eq:closure-root} is invariant to discovery order, resolver-local
node-handle alpha-renaming, one-for-one substitution of an alias for the same
interned canonical node while the canonical edge multiset is fixed, path
spelling that canonicalizes to the same locator, and enumeration order.
\end{lemma}

\begin{proof}
Kind-specific canonicalization is total on admitted locators. Local handles are
quotiented by \(h_V\), and SCC condensation over the resulting canonical-node
graph is unique up to component naming. Every node, internal edge, component,
and inter-component edge is domain separated and sorted before hashing.
Therefore, handle alpha-renaming, one-for-one interned-alias substitution,
permutations, and equivalent admitted spellings yield identical byte encodings
at every level when the canonical edge multiset is unchanged.
The explicit \(h_V(r_s)\) distinguishes canonical roots inside one SCC;
the endpoint-sensitive \(h_C\) distinguishes different cross-component
attachments. The statement does not identify semantically equivalent but
differently evidenced programs.
\end{proof}

\subsection{Task grants and policy order}

A trusted task grant is

\[
\gamma=\langle p,\tau,\iota,\psi,R,D,A,\rho_\gamma,
\kappa,\mathbf e,\pi_\gamma,c_\gamma,[t_{\mathrm{nb}},t_{\mathrm{exp}}),n\rangle .
\]

Here \(p\) is principal, \(\tau\) tenant, \(\iota\) a version-qualified exact
task atom, \(\psi\) purpose, \(R\) recipients, \(D\) resource/data bounds,
\(A\) an effect ceiling, \(\rho_\gamma\) the authorized closure root,
\(\kappa\) the effect-contract version, \(\mathbf e\) expected exact epochs,
\(\pi_\gamma\) the evidence-profile identifier, \(c_\gamma\) the
canonicalizer/domain version, \([t_{\mathrm{nb}},t_{\mathrm{exp}})\) validity,
and \(n\) a nonce. The authenticated \(A\) is an immutable issuer-compiled
ceiling; later policy expansion cannot enlarge it. Grant construction fixes
\(B=\gamma.A\) before closure discovery and uses that same bound for issuance,
prepare, and commit resolution.

An effect policy denotes a set of external transitions. For policies
\(a_1,a_2\),

\[
a_1\restrict a_2
\iff
\denote{a_1}\subseteq\denote{a_2}.
\]

Lower policies authorize no more behavior. Let \(C_v\) be the authenticated
effect policy carried by an applicable contract node. For the current effect,
the closure operand is the relational meet

\[
A_{\mathrm{closure}}(\Graph_t,e)=
\bigotimes_{\substack{v\in\mathcal I_{\Graph_t}(e)\\
                      C_v\ \mathrm{applies\ to}\ e}} C_v .
\]
Here \(\bigotimes\) denotes iterated relational meet over the indexed policy
operands; Appendix~\ref{app:proofs} gives its binary definition and denotational
law.

The resolver obtains \(\mathcal I_{\Graph_t}(e)\) from the authenticated
root-to-mediated-sink provenance emitted by the adapter; under
Assumption~\ref{ass:closure}, it includes every contract capable of influencing
that event. If this provenance or a required contract is unresolved, the meet
is undefined and evaluation fails closed. Current task authority, current
closure contracts, runtime state, and purpose/provenance constraints compose
by intersection, not by field-wise union:

\begin{equation}
A_{\mathrm{cur}}(t,e)=A_{\mathrm{task}}(t)\cap
A_{\mathrm{closure}}(\Graph_t,e)\cap A_{\mathrm{state}}(t)
\cap A_{\mathrm{purpose/prov}}(t,e),
\label{eq:effective}
\end{equation}

and the final ceiling is
\(\denote{\gamma.A}\cap\denote{A_{\mathrm{cur}}(t,e)}\).
Metadata supplied by a Skill may propose a narrowing contract to a trusted
compiler, but cannot modify authenticated task/organizational operands or the
already issued \(\gamma.A\).

\subsection{Normalized external effects}

For a surface event \(x\) in current state \(S_t\), a trusted adapter produces

\[
\begin{aligned}
e_t=\Norm(x,S_t)=\langle&
op,target,recipient,dataClass,amount,argsDigest,\\
&preStateDigest,postconditionContract,purpose,taskAtom,\\
&effectContractVersion,durability,idempotencyScope\rangle .
\end{aligned}
\]

\(\Norm\) is partial. Failure, ambiguity, or an unsupported surface returns
\(\bot\), which denies. We call \(e_t\) the \emph{realized effect event} because
its concrete surface, arguments, target, and authoritative pre-state are fixed
at the durability gate; the term does not assert that the receiver's
postcondition has already occurred or is truthful.

Let \(\Omega\) be a declared external observation map.

\begin{definition}[Effect equivalence]
\label{def:effect-equivalence}
\[
e_1\equiv_\Omega e_2
\iff
\forall S\in\Sigma_\Omega:
\Omega(\denote{e_1}(S))=\Omega(\denote{e_2}(S)).
\]
\(\Sigma_\Omega\) is the common set of admissible pre-states on which both
transition denotations and the post-state observation map are defined.
\end{definition}

\begin{definition}[Registered policy-selection congruence]
\label{def:policy-congruence}
Let \(\mathcal B(\gamma,\Graph,S,t,e)\) be the canonical multiset of authenticated
policy operands selected for (A4), including the immutable grant ceiling,
current task/state/purpose operands, and every applicable contract selected
through \(\mathcal I_{\Graph}(e)\). An adapter group is policy-selection
congruent iff, for every supported pair \(e_1\equiv_\Omega e_2\) evaluated in
the same authenticated \((\gamma,\Graph,S,t)\),
\[
\begin{split}
\operatorname{authKey}(e_1)&=\operatorname{authKey}(e_2),\\
\denote{\bigotimes\mathcal B(\gamma,\Graph,S,t,e_1)}
&=\denote{\bigotimes\mathcal B(\gamma,\Graph,S,t,e_2)} .
\end{split}
\]
Here \(\operatorname{authKey}\) contains every normalized effect field consumed
by (A1)--(A6) but excludes the execution-surface name.
\end{definition}

The monitor need not prove semantic equivalence for arbitrary code. Registry
admission establishes Definition~\ref{def:policy-congruence}, normally by
requiring supported equivalent surfaces to emit the same canonical IR and by
making influence/contract selection a deterministic function of that IR and
the authenticated closure. A conservative adapter may emit an
over-approximated effect and thereby deny.

\subsection{Joint operation witness}

A prepared reservation at \(t_p\) binds the immutable record

\[
W_{t_p}=\langle H(\gamma),\rho_\gamma,\rho_p,H(e_p),
H(\lambda_p),\mathbf e_p,H(S_{t_p}^{pre}),\pi_\gamma,c_\gamma,r,n,t_p\rangle ,
\]

where \(\lambda_p\) is the purpose/provenance/recipient label and \(r\) is a
single-consumption reservation. Let \(\operatorname{Rec}(r)\) be the
authenticated reservation-store record. We define

\begin{align*}
\operatorname{BindingIntact}(W_{t_p},\gamma)\iff{}&
W.\mcode{grant}=H(\gamma)\\
&{}\land W.\mcode{grantRoot}=\gamma.\rho_\gamma
\\
&{}\land W.\mcode{prepareRoot}=\gamma.\rho_\gamma
\land W.\mcode{epochs}=\gamma.\mathbf e\\
&{}\land W.\mcode{profile}=\gamma.\pi_\gamma
\land W.\mcode{canon}=\gamma.c_\gamma\\
&{}\land W.\mcode{nonce}=\gamma.n\\
&{}\land \operatorname{Rec}(r)=
\langle H(\operatorname{enc}(W_{t_p}\setminus\{r\})),
\mcode{fresh}\rangle .
\end{align*}

Thus a field from another grant, profile, canonicalizer, effect, state, or
reservation changes the authenticated serialization rather than merely a local
variable.

\begin{definition}[Commit admissibility]
\label{def:admissible}
\(\operatorname{Admit}(\gamma,W_{t_p},e_{t_c},S_{t_c})\) holds at the commit
linearization time \(t_c\) iff:

\begin{align}
&\operatorname{Authn}(\gamma)\land
t_{\mathrm{nb}}\le t_c<t_{\mathrm{exp}}
\land \operatorname{TaskCurrent}(\iota,p,\tau,t_c), \tag{A1}\\
&\operatorname{Exact}(\Graph_{t_c}(s;\pi_\gamma,c_\gamma,\gamma.A),\pi_\gamma)
\land\rho_{t_c}(s;\pi_\gamma,c_\gamma,\gamma.A)
=W.\mcode{prepareRoot}=\rho_\gamma, \tag{A2}\\
&\operatorname{EpochsCurrent}(\mathbf e_{t_c},W.\mcode{epochs})
\land\operatorname{ContractCurrent}(\kappa,t_c)
\land\operatorname{ProfileCurrent}(\pi_\gamma,c_\gamma,t_c), \tag{A3}\\
&e_{t_c}\in\denote{\gamma.A}
\land e_{t_c}\in\denote{A_{\mathrm{cur}}(t_c,e_{t_c})}
\land\operatorname{LabelFits}(\lambda_{t_c},\psi,R,D), \tag{A4}\\
&H(S_{t_c}^{pre})=W.\mcode{pre}
\land\operatorname{FreshReservation}(r,n,t_c), \tag{A5}\\
&H(e_{t_c})=W.\mcode{effect}\land
H(\lambda_{t_c})=W.\mcode{label}\land
\operatorname{BindingIntact}(W_{t_p},\gamma). \tag{A6}
\end{align}
\end{definition}

The use of one witness is material: roots from one grant, effect ceilings from
another, or fresh epochs from a third may not be recombined into an authority
that no issuer granted.

\subsection{Safety invariants}

For every admitted commit event \(c\):

\begin{description}[leftmargin=1.1em]
  \item[I1 --- Trusted authority.] \(c\)'s grant descends from the authenticated
  current-task authority; Skill metadata is not an authority root.
  \item[I2 --- Exact closure.]
  \(c.\rho=\rho_t(s;\pi_\gamma,c_\gamma,\gamma.A)=\rho_\gamma\), and every
  material node and edge satisfies the grant-bound evidence profile.
  \item[I3 --- Non-amplified effect.] \(c.e\in\denote{\gamma.A}\cap
  \denote{A_{\mathrm{cur}}(t,c.e)}\).
  \item[I4 --- Current operation.] All exact epochs, contract versions,
  pre-state, and reservation facts hold at the linearization point.
  \item[I5 --- Path invariance.] Supported effect-equivalent surface events are
  checked against the same policy denotation.
\end{description}

\section{ClosureBound Design}
\label{sec:design}

Figure~\ref{fig:architecture} shows the control path. Blue components produce
trusted task authority, purple components resolve identity, green components
mediate effects, amber denotes reauthorization or uncertainty, and red denotes
fail-closed rejection. Supply-chain attestations and Skill scanners feed
evidence; they do not bypass the task/effect decision.

\begin{figure}[t]
\centering
\begin{tikzpicture}[
  font=\scriptsize,
  >=Latex,
  node distance=4.5mm and 5mm,
  box/.style={rounded corners=2pt, draw, minimum height=7.2mm,
              align=center, inner xsep=4pt, inner ysep=2.5pt},
  flow/.style={->,line width=.65pt},
  gate/.style={diamond,aspect=2.15,draw,align=center,inner sep=1.7pt}
]
\node[box,draw=opBlue,fill=opBlueFill,text=opBlueText] (task)
  {trusted task\\grant issuer};
\node[box,draw=opPurple,fill=opPurpleFill,text=opPurpleText,right=of task] (resolver)
  {typed recursive\\closure resolver};
\node[box,draw=opPurple,fill=opPurpleFill,text=opPurpleText,right=of resolver] (root)
  {canonical graph\\commitment \(\rho\)};
\node[box,draw=opGreen,fill=opGreenFill,text=opGreenText,below=8mm of resolver] (norm)
  {surface adapter +\\effect normalizer};
\node[box,draw=opBlue,fill=opBlueFill,text=opBlueText,left=of norm] (state)
  {current state,\\epochs, labels};
\node[gate,draw=opGreen,fill=opGreenFill,text=opGreenText,right=of norm] (joint)
  {joint\\predicate};
\node[box,draw=opGreen,fill=opGreenFill,text=opGreenText,right=10mm of joint] (commit)
  {atomic reserve/\\commit adapter};
\node[box,draw=opAmber,fill=opAmberFill,text=opAmberText,below=of joint] (reauth)
  {typed contract/\\evidence diff +\\new authorization};
\node[box,draw=opRed,fill=opRedFill,text=opRedText,below=of commit] (deny)
  {deny or\\indeterminate};
\draw[flow,opBlue] (task) -- (resolver);
\draw[flow,opPurple] (resolver) -- (root);
\draw[flow,opPurple] (root.south) -- (joint.north);
\draw[flow,opBlue] (task.south) -- ++(0,-4mm) -| (joint.135);
\draw[flow,opBlue] (state) -- (norm);
\draw[flow,opGreen] (norm) -- (joint);
\draw[flow,opGreen] (joint) --
  node[midway,above=2.5mm]{all current} (commit);
\draw[flow,opAmber] (joint) -- node[midway,left=1mm]{changed} (reauth);
\draw[flow,opRed] (joint.south east) to[out=-18,in=180]
  node[pos=.68,below=1mm,align=center]{invalid/\\unknown} (deny.west);
\draw[flow,opAmber,dashed] (reauth.west) -|
  ([xshift=-4mm]state.west) |- (task.west);
\draw[flow,opRed,dashed] (commit) -- node[right]{uncertain settlement} (deny);
\end{tikzpicture}
\caption{\ClosureBound architecture. Authority, exact transitive identity,
current state, purpose/provenance, and the normalized effect meet at one
operation-time decision. A material change creates a new grant; it never edits
the old closure binding.}
\Description{A pipeline with trusted task grant, closure resolver, graph root,
current state, effect normalizer, joint predicate, atomic commit, reauthorization,
and deny or indeterminate outcomes.}
\label{fig:architecture}
\end{figure}
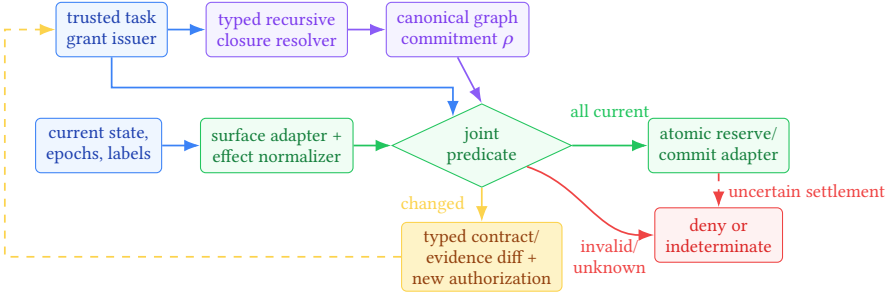

\subsection{Resolver adapters and evidence profiles}

Different node kinds cannot share one naïve notion of version. ClosureBound
uses a small adapter interface:

\[
\mathsf{resolve}_k(locator,parent,context,\pi,c_\pi,B,t)
\rightarrow
\langle q_t(v),outgoing,status,contextKey\rangle .
\]

An outgoing item is
\(\langle r,k',locator',q^E,context'\rangle\). The context key includes every
lockfile, feature selection, platform, tenant, and resolver-policy fact that can
change outgoing dependencies. Its authenticated digest is the \(\zeta\) field
of the contextual vertex identity. Consequently, interning and memoization by
\((q,contextKey)\), rather than by \(q\) alone, cannot collapse two
context-sensitive closures or erase which outgoing subgraph belongs to which
call context.

\begin{itemize}
  \item A file adapter binds the canonical path, bytes, executable mode, and
  enclosing root; symlink policy is explicit.
  \item A package adapter binds ecosystem, registry namespace, immutable
  version/digest, lock evidence, and resolved runtime dependencies.
  \item A recursive Skill adapter resolves the target directory and repeats the
  Skill/file/package/service walk.
  \item A tool adapter binds the server identity, method schema, argument and
  result contracts, and declared effect contract.
  \item A model adapter binds an immutable artifact or a provider deployment
  identity at the evidence strength required by policy.
  \item A service adapter binds canonical audience, authenticated contract,
  deployment/epoch evidence where available, and the effect adapter version.
\end{itemize}

The MCP 2026-07-28 revision makes this context requirement concrete. Its
stateless core carries protocol version and client capabilities per request,
adds server discovery, and makes tool/resource listings deterministic and
cacheable with \mcode{ttlMs} and \mcode{cacheScope}; authorization credentials
are issuer-bound~\cite{mcp2026changelog,mcp2026auth}. An MCP adapter therefore
binds at least the negotiated protocol revision, authenticated server identity,
tool schema/listing digest, cache scope and observed freshness interval, plus
any explicit server-minted cross-call handle. A TTL is a cache hint, not proof
that an old grant authorizes a later deployment or effect.

An evidence profile \(\pi\) sets minimum classes per node/effect risk. For
example, a read-only documentation task may admit a service URL plus schema
epoch, while a payment task may require an attested deployment and
transactional idempotency contract. Profiles are policy, not self-declared
Skill metadata.

\subsection{Recursive resolution}

Algorithm~\ref{alg:resolve} builds the closure. Its context-sensitive visited
set terminates repeated expansion; SCC canonicalization handles cycles after
the walk. Discovery order is recorded for diagnosis but excluded from the
canonical root.

\begin{algorithm}[t]
\caption{\textsc{Resolve-Closure}}
\label{alg:resolve}
\begin{algorithmic}[1]
\Require root Skill \(s\), profile \(\pi\), canonicalizer \(c_\pi\), candidate
         ceiling \(B\), trusted time \(t\)
\Ensure exact graph \(G\) and root \(\rho\), or fail closed
\State \(Q\gets[(\mathsf{skill},s,\bot,\bot,\bot,ctx_{\mathrm{root}})]\);
       \(V\gets\emptyset\); \(E\gets\emptyset\); \(Seen\gets\emptyset\)
\While{\(Q\ne\emptyset\)}
  \State \((k,\ell,parent,r_{in},q^E_{in},ctx)\gets\Call{Pop}{Q}\)
  \State \(\ell'\gets C^{c_\pi}_k(\ell,parent,ctx)\);
         reject if canonicalization fails
  \State \((q,outs,status,ctxKey)\gets
         \mathsf{resolve}_k(\ell',parent,ctx,\pi,c_\pi,B,t)\)
  \If{\(status\ne\mcode{exact}\) or \(q.\sigma\not\models\pi(k)\)}
    \State \Return \(\mcode{UNRESOLVED}(k,\ell',status)\)
  \EndIf
  \State verify \(q.\zeta=H(\operatorname{enc}(ctxKey))\) under \(\pi\)
  \State \(v\gets\Call{Intern}{q,ctxKey}\); \(V\gets V\cup\{v\}\)
  \If{\(parent=\bot\)} \(r_s\gets v\) \EndIf
  \If{\(parent\ne\bot\)}
    \State verify \(q^E_{in}\) authenticates \((parent,r_{in},v)\) under \(\pi\)
    \State \(E\gets E\cup\{(parent,r_{in},v,q^E_{in})\}\)
  \EndIf
  \If{\((v,ctxKey)\notin Seen\)}
    \State \(Seen\gets Seen\cup\{(v,ctxKey)\}\)
    \ForAll{\((r,k',\ell_o,q^E,ctx_o)\in outs\)}
      \State \(\Call{Push}{Q,(k',\ell_o,v,r,q^E,ctx_o)}\)
    \EndFor
  \EndIf
\EndWhile
\State reject if any node/edge material under \((\pi,c_\pi,B)\) is unreachable,
       missing, stale, or \mcode{NOASSERTION}
\State \(G\gets(V,E,r_s)\);
       \(\rho\gets\Call{SCC-Merkle-Commit}{G,\pi,c_\pi,B}\)
\State \Return \((G,\rho)\)
\end{algorithmic}
\end{algorithm}

Resolver output is evidence, not authority. A malicious but immutably pinned
package will be identified consistently and can still be denied by the
grant/effect policy.

\subsection{Grant compilation}

Before closure discovery, the grant compiler derives a trusted upper bound
\(B_0\) from the authenticated current-task certificate, organizational/user
policy, and purpose constraints. A Skill-authored manifest may propose an
additional narrowing conjunct, but cannot add an effect absent from those
trusted operands. Closure-derived contracts do not contribute to \(B_0\).
The resolver then computes \(\Graph_t(s;\pi,c_\pi,B_0)\); the compiler validates
that its prospective relational clauses are representable and freezes
\(\gamma.A=B_0\) together with \(\rho_\gamma\). At prepare and commit,
closure-derived contracts remain a separately recomputed conjunct of
\(A_{\mathrm{cur}}\) in Equation~\eqref{eq:effective}; they never feed a
narrowed result back into the graph-discovery bound. This two-phase rule avoids
a circular equation of the form \(A\leftarrow\Graph(A)\).

Cross-dimensional correlation is preserved. A grant for
\((\mathsf{billing},\mathsf{vendorA},\$500)\) and another for
\((\mathsf{refund},\mathsf{customerB},\$50)\) cannot donate their fields to
authorize \((\mathsf{billing},\mathsf{customerB},\$500)\). The executable
representation is a finite relation of complete effect clauses: authorization
meets preserve whole tuples rather than unioning independently permitted
fields. Appendix~\ref{app:proofs} gives the relational-meet argument and the
corresponding clause-splicing counterexample.

\subsection{Effect normalization and path registry}

Each protected adapter registers:

\[
\alpha=\langle surfaceVersion,match,normalize,effectContract,
stateWitness,durability,fence,settle\rangle .
\]

\mcode{match} identifies intercepted events without trusting model prose;
\mcode{normalize} emits the IR; \mcode{stateWitness} reads authoritative
preconditions; \mcode{durability} declares the irreversible boundary; and
\mcode{fence} connects the monitor reservation to the sink.

The registry groups adapters by effect-contract denotation. A direct payment
API and a browser payment form may differ in fields and timing but normalize to
the same payment clause. A shell adapter that cannot determine the payee and
amount before execution cannot join that group; it must be denied or confined
in a sandbox whose maximum possible effect is authorized.

\subsection{Prepare and commit}

Algorithm~\ref{alg:commit} makes the temporal boundary explicit.
\mcode{Prepare} is advisory and reserves a candidate. \mcode{Commit} repeats
all volatile checks and consumes a fence immediately before durability.
Define \(\operatorname{PrepareEligible}(\gamma,\xi_p,G_p,\rho_p,S_p,
\lambda_p,\mathbf e_p,e_p,t_p)\) as the prospective form of (A1)--(A4):
every occurrence of a prepared-witness field is replaced by its value in
\(\xi_p\), while reservation-only (A5)--(A6) are not yet evaluated. Thus this
predicate is defined before \(W_{t_p}\) exists.

\begin{algorithm}[t]
\caption{\textsc{Prepare-Then-Commit}}
\label{alg:commit}
\begin{algorithmic}[1]
\Require authenticated grant \(\gamma\), root Skill \(s\), surface event \(x\)
\State \(t_p\gets\Call{Trusted-Time}{}\)
\State \((\xi_p,G_p,\rho_p,S_p,\lambda_p,\mathbf e_p,e_p)\gets
       \Call{Acquire-Consistent-Snapshot}
       {s,x,\gamma.\pi_\gamma,\gamma.c_\gamma,\gamma.A,t_p}\)
\State deny unless \(e_p\ne\bot\) and
       \(\operatorname{PrepareEligible}(\gamma,\xi_p,G_p,\rho_p,S_p,
       \lambda_p,\mathbf e_p,e_p,t_p)\)
\State \(W^{-}\gets\langle H(\gamma),\rho_\gamma,\rho_p,H(e_p),
       H(\lambda_p),\mathbf e_p,H(S_p),\pi_\gamma,c_\gamma,n,t_p\rangle\)
\State \(r\gets\Call{Reserve}{H(\operatorname{enc}(W^{-})),n}\);
       \(W_{t_p}\gets W^{-}\cup\{r\}\)
\Statex
\State \Comment{Operation-time linearization}
\State \(t_c\gets\Call{Trusted-Time}{}\)
\State \((\xi_c,G_c,\rho_c,S_c,\lambda_c,\mathbf e_c,e_c)\gets
       \Call{Acquire-Consistent-Snapshot}
       {s,x,\gamma.\pi_\gamma,\gamma.c_\gamma,\gamma.A,t_c}\)
\State \((d,f)\gets\Call{Atomic-Validate-Consume}
       {\gamma,W_{t_p},\xi_c,\rho_c,H(S_c),H(\lambda_c),
        \mathbf e_c,H(e_c),t_c}\)
\If{\(d\ne\mcode{ALLOW\_COMMIT}\)}
  \State \Return \(d\) with authenticated failure evidence
  \Statex \hspace{\algorithmicindent}and exact typed diff when available;
          otherwise \(\mathsf{UNKNOWN}\)
\EndIf
\State \(\omega\gets\Call{Dispatch-With-Fence}{x,e_c,f}\)
\If{authenticated settlement or independent postcondition verifies}
  \State \Return \mcode{COMMITTED}
\ElsIf{dispatch is known not to have crossed durability}
  \State \Return \mcode{ABORTED}
\Else
  \State \Return \mcode{INDETERMINATE}
\EndIf
\end{algorithmic}
\end{algorithm}

\mcode{Acquire-Consistent-Snapshot} binds the resolver frontier, state,
labels, epochs, effect, profile, and canonicalizer in the opaque token
\(\xi\). \mcode{Atomic-Validate-Consume} re-evaluates (A1)--(A6), verifies that
all supplied values still equal the authoritative facts represented by
\(\xi_c\), and consumes \(r\) in the same compare-and-swap. The returned fence
is an authenticated one-shot token over at least
\((H(\gamma),\rho_c,H(e_c),H(S_c),\mathbf e_c,n,r)\). A trusted adapter or
receiver must re-normalize the actual request and atomically require the same
\(H(e_c)\) before crossing durability; mutable \(x\) is never used as an
unbound authority input. A local database transaction can implement this
directly. A remote service can consume the effect-bound fence and reject older
frontier/epoch tokens. If neither pattern is available, the system can
authorize dispatch but cannot claim atomic operation-time safety for the
receiver.

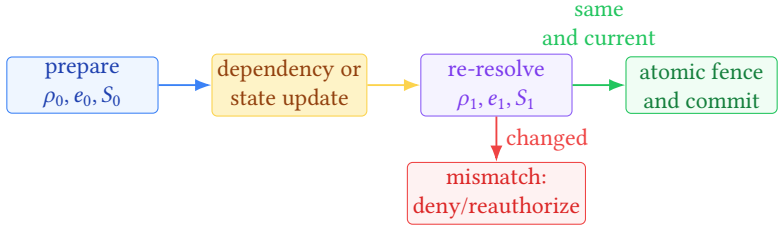
\begin{figure}[t]
\centering
\begin{tikzpicture}[
  font=\small,
  >=Latex,
  phasebox/.style={rounded corners=2pt,draw,minimum width=20mm,minimum height=7mm,
               align=center,inner sep=2pt},
  arr/.style={->,line width=.7pt}
]
\node[phasebox,draw=opBlue,fill=opBlueFill,text=opBlueText] (p)
  {prepare\\\(\rho_0,e_0,S_0\)};
\node[phasebox,draw=opAmber,fill=opAmberFill,text=opAmberText,right=7mm of p] (u)
  {dependency or\\state update};
\node[phasebox,draw=opPurple,fill=opPurpleFill,text=opPurpleText,right=7mm of u] (r)
  {re-resolve\\\(\rho_1,e_1,S_1\)};
\node[phasebox,draw=opRed,fill=opRedFill,text=opRedText,below=6mm of r] (d)
  {mismatch:\\deny/reauthorize};
\node[phasebox,draw=opGreen,fill=opGreenFill,text=opGreenText,right=7mm of r] (c)
  {atomic fence\\and commit};
\draw[arr,opBlue] (p) -- (u);
\draw[arr,opAmber] (u) -- (r);
\draw[arr,opRed] (r) -- node[right]{changed} (d);
\draw[arr,opGreen] (r) --
  node[midway,above=4mm,align=center]{same\\and current} (c);
\end{tikzpicture}
\caption{Operation-time update discipline. Preparation never makes an old
closure authoritative. A distinguishable material mismatch at the
linearization point forces denial or a new grant.}
\Description{Timeline from prepare, through dependency update and re-resolution,
to either deny and reauthorize or atomic commit.}
\label{fig:timeline}
\end{figure}

\subsection{Typed contract/evidence diff and reauthorization}

When \(\rho_1\ne\rho_\gamma\), the system reports a typed diff:

\[
\Delta=\langle V^+,V^-,V^\sim,E^+,E^-,
A^+,A^-,\textsf{evidenceChanges}\rangle ,
\]

where \(V^\sim\) records version/digest/epoch changes and \(A^+,A^-\) are
differences in the finite adapter-declared effect-contract denotations, not
inferences of arbitrary program semantics. If the authenticated contracts do
not make that difference exactly representable, the field is
\(\mathsf{UNKNOWN}\) and review fails closed. Cosmetic metadata outside the
material closure can be marked non-authoritative, but the old grant remains
immutable. A trusted issuer creates \(\gamma'\) over \(\rho_1\) after reviewing
\(\Delta\). This transition is widening/renewal, not ordinary execution.

\subsection{Caching without stale authority}

Resolvers may cache immutable nodes by digest. Mutable lookups are keyed by
authority epoch and evidence profile. A cached graph is a performance hint:
commit still validates an authenticated root-to-currentness certificate or
re-resolves. Cache age is not a security proof. Negative and unresolved results
may be cached only fail closed.

\section{Security Analysis}
\label{sec:security}

\subsection{Metadata cannot mint authority}

\begin{theorem}[Metadata non-authority]
\label{thm:metadata}
Under Assumptions~\ref{ass:closure}, \ref{ass:fresh}, \ref{ass:linear},
and~\ref{ass:crypto} and the grant-compilation rule, changing Skill-authored
metadata, instructions, self-declared permissions, or explanations cannot add
an element to the effect set admitted by an existing grant.
\end{theorem}

\begin{proof}
The issuer authenticates and freezes \(\gamma.A\); Skill data cannot rewrite
that field. At commit, (A4) requires membership both in this immutable ceiling
and in the independently recomputed current task, closure, state, and
purpose/provenance meet. A proposed Skill operand is accepted only through the
trusted compiler as an additional conjunct. Enlarging or replacing that
proposal still cannot authorize an effect outside \(\gamma.A\). A changed
material instruction also changes \(\rho\), so the old grant fails (A2). The
result does not say metadata cannot influence model behavior; it says such
influence is not trusted authority at the mediated effect boundary.
\end{proof}

\subsection{Closure and version non-inheritance}

\begin{lemma}[Material-change separation]
\label{lem:separation}
Under Assumption~\ref{ass:crypto}, changing the canonical root identity so that
\(h_V(r_s)\) changes, a canonical material node record, an authenticated typed
material-edge record (including either endpoint's canonical \(h_V\)), or SCC
topology changes \(\rho\) except with negligible cryptographic probability.
\end{lemma}

\begin{proof}
The changed record changes a domain-separated node, edge, SCC, or
endpoint-sensitive component-edge encoding. An SCC-internal root change alters
the separately committed \(h_V(r_s)\). Each change propagates to the final
root. Equality would require a collision or second preimage in \(H\). Changes
outside the declared material closure are not covered.
\end{proof}

\begin{theorem}[Version/closure non-inheritance]
\label{thm:noninherit}
Under Assumptions~\ref{ass:closure}, \ref{ass:fresh}, \ref{ass:linear},
and~\ref{ass:crypto}, let \(\gamma\) authorize closure root \(\rho_0\). If, at
the commit linearization point, a material transitive dependency,
service/tool/effect contract, or typed edge has a canonical value different
from the value bound by \(\rho_0\), \(\gamma\) cannot admit that changed closure
without a new grant.
\end{theorem}

\begin{proof}
By Lemma~\ref{lem:separation}, an exactly resolved material change gives
\(\rho_1\ne\rho_0\) except with negligible probability. If resolution is no
longer exact, Definition~\ref{def:exact} fails. In the first case (A2) rejects
the mismatch; in the second it rejects incompleteness. Algorithm~\ref{alg:commit}
returns \mcode{REAUTHORIZE} or \mcode{DENY}, and the immutable grant cannot be
edited.
\end{proof}

\statementheading{Non-conclusions}
The theorem does not detect behavioral changes behind unchanged weak service
evidence, malicious behavior already present in \(\rho_0\), or omitted
dependencies that violate Assumption~\ref{ass:closure}.

\subsection{Transitive effect non-amplification}

\begin{lemma}[Monitor refinement]
\label{lem:refinement}
If the atomic validation step in Algorithm~\ref{alg:commit} returns its
commit-allow decision---equivalently, if execution reaches the fenced dispatch
step---then
\(\operatorname{Admit}(\gamma,W_{t_p},e_{t_c},S_{t_c})\) held at that
validation step's linearization point.
\end{lemma}

\begin{proof}
The only path to dispatch requires \mcode{Atomic-Validate-Consume} to return
\mcode{ALLOW\_COMMIT}. By its interface contract and
Assumption~\ref{ass:linear}, that atomic transition evaluates (A1)--(A6) over
the authoritative snapshot \(\xi_c\), verifies the full frontier/profile/state/
effect tuple, and changes the matching reservation from fresh to consumed.
Every non-admitting result exits before dispatch. The later
\mcode{COMMITTED}, \mcode{ABORTED}, and \mcode{INDETERMINATE} labels report
settlement knowledge and cannot retroactively change the admission predicate;
in particular, an indeterminate dispatch still passed the same authorization
linearization point.
\end{proof}

\begin{theorem}[Transitive effect non-amplification]
\label{thm:nonamplification}
Under Assumptions~\ref{ass:closure}--\ref{ass:linear}
and~\ref{ass:crypto}, every effect \(e_{t_c}\) admitted by ClosureBound at its
linearization time \(t_c\) is in
\[
\denote{\gamma.A}\cap\denote{A_{\mathrm{task}}(t_c)}\cap
\denote{A_{\mathrm{closure}}(\Graph_{t_c},e_{t_c})}\cap
\denote{A_{\mathrm{state}}(t_c)}\cap
\denote{A_{\mathrm{purpose/prov}}(t_c,e_{t_c})} .
\]
No ordinary dependency path adds authority beyond this intersection.
\end{theorem}

\begin{proof}
By Lemma~\ref{lem:refinement}, every effect released to
\mcode{Dispatch-With-Fence} satisfied
Definition~\ref{def:admissible}. Complete mediation gives its normalized
\(e_{t_c}\). Complete closure and (A2) bind the exact current implementation,
authenticated influence paths, and applicable effect contracts. (A1) binds
trusted current-task authority; (A3) binds current epochs/contracts/profile;
(A4) separately checks the immutable grant ceiling and current relational meet;
(A5)--(A6) bind state, labels, reservation, and effect to the same prepared
witness. Hence the effect is in every displayed operand. A nested Skill,
package, or service is already within the bound closure and cannot supply a
separate authority operand.
\end{proof}

\statementheading{Non-conclusions}
The result is about represented mediated effects. It does not prevent covert
channels, prove source truth, or repair a compromised normalizer or receiver.

\subsection{Operation-time TOCTOU safety}

\begin{theorem}[Operation-time freshness]
\label{thm:toctou}
Under Assumptions~\ref{ass:crypto}, \ref{ass:fresh}, and~\ref{ass:linear}, a
bound closure root, expected epoch, effect contract, normalized effect, or
pre-state whose authoritative value at the commit linearization point differs
from the prepared/grant-bound value cannot be admitted using the prepared
reservation.
\end{theorem}

\begin{proof}
Algorithm~\ref{alg:commit} reacquires every volatile fact at commit. A closure,
effect, label, epoch, contract, profile, canonicalizer, or pre-state change
that remains distinguishable at the linearization point violates (A2)--(A6).
Transient ABA histories that authoritative evidence intentionally maps back to
the same identity are outside this value-freshness claim; deployments requiring
ABA detection must use non-reused epochs. The opaque snapshot token and
\mcode{Atomic-Validate-Consume} compare the complete current tuple and consume
the reservation in one atomic transition, preventing a concurrent update from
validating against old facts after the check. Without atomicity or fencing, the
theorem's premise does not hold.
\end{proof}

\subsection{Equivalent-effect path invariance}

\begin{theorem}[Supported path invariance]
\label{thm:path}
Assume effect mediation is surface-complete, adapters are extensionally sound
with respect to \(\Omega\), and their registry group satisfies
Definition~\ref{def:policy-congruence}. If \(x_1,x_2\) are supported surface
events evaluated under the same authenticated grant, closure, state, and
linearization time, and
\(\Norm(x_1,S)\equiv_\Omega\Norm(x_2,S)\), ClosureBound evaluates them against
the same effective policy denotation and neither path can bypass the
closure/currentness predicates.
\end{theorem}

\begin{proof}
Both surfaces enter Definition~\ref{def:admissible}. Registered
policy-selection congruence gives the same authorization key and the same
denotation for the complete applicable policy bundle, including the
effect-dependent influence/contract selector. Conditions (A1)--(A3) and
(A5)--(A6) are enforced for every registered path and bind the root, epochs,
state, reservation, profile, labels, and normalized effect. Therefore
switching supported paths cannot enlarge the authorization set.
\end{proof}

\statementheading{Non-conclusions}
ClosureBound does not discover arbitrary semantic equivalence. The theorem
excludes an unregistered sink and fails closed when a general shell/browser
event cannot be normalized before durability.

\subsection{Necessity counterexamples}

The conjunction is not cosmetic. Each reduced mechanism admits a short
counterexample:

\begin{description}[leftmargin=1.2em]
  \item[Root-only identity.] Keep root bytes fixed and update a nested package
  or service deployment.
  \item[Closure only at load.] Prepare under \(\rho_0\), update a helper, and
  commit without re-resolution.
  \item[Effect-only commit.] Keep the payment IR within its ceiling but change
  the implementation to an unauthorized analytics-bearing closure.
  \item[Closure-only commit.] Keep the closure fixed and switch the recipient
  or amount through a mutable event.
  \item[Surface-name policy.] Deny a direct payment tool but reach the same
  payment through browser or shell.
  \item[Uncorrelated field union.] Combine purpose from one clause, recipient
  from another, and amount from a third.
\end{description}

\begin{proposition}[No represented reduced mechanism family is sufficient]
\label{prop:necessity}
For each mechanism family removed by a reduced profile in
Table~\ref{tab:profiles}, and for each single-point mutant in
Table~\ref{tab:mutants}, there exists a finite trace on which that weakened
evaluator admits an effect excluded by the full joint predicate.
\end{proposition}

\begin{proof}
The six constructions above cover the profile-level mechanism families.
Replay, grant forgery/expiry, and stale pre-state cases exercise additional
atomic conditions. Each seeded mutant changes exactly one implementation
mechanism while retaining the full-profile evaluator elsewhere; the artifact
executes and records the corresponding trace. The traces establish
mechanism-family necessity; per-condition independence across (A1)--(A6) and
exhaustive attack enumeration are separate questions.
\end{proof}

\section{Reference Implementation}
\label{sec:implementation}

\subsection{Scope}

The artifact is a deterministic Node.js reference kernel with no hosted model,
network dependency, or private data. It implements canonical graph
commitments, finite effect clauses, epoch checks, fixture execution, mutants,
bounded exploration, aggregation, and consistency
verification. Reservation freshness and single consumption are modeled as
evaluator facts and lifecycle transitions; the artifact does not implement a
production \mcode{Reserve}/\mcode{Atomic-Validate-Consume} transaction.

The reference model is a publication-oriented research implementation over
synthetic records, separate from Accentrust production systems. It contains no
private adapter, customer or tenant record, production endpoint or schema,
credential, deployment topology, operational policy, or incident trace.

The ClosureBound evaluation implementation and its machine-readable execution
traces are retained internally and are not included in the arXiv source bundle.
Appendix~\ref{app:artifact} states the disclosed configuration and resulting
availability boundary.

The reference profile abstracts production signature verification, package
resolution, MCP servers, browser drivers, and distributed fencing.
Cryptographic records use deterministic SHA-256 digests and explicit
\mcode{authenticated} flags. This design isolates binding logic; key management
and provider security are outside the evaluated properties.

\subsection{Canonical records}

Node records use a fixed JSON canonicalizer with lexicographically sorted
object keys, explicit array-order rules, domain tags, and rejection of
non-finite numbers or unsupported values. Graph arrays are sorted only after
node and edge identity has been fixed. Resolver handles are excluded from node
digests; authenticated resolution-context digests remain included. Duplicate
canonical nodes are interned; edge multiplicity and kind remain explicit.

The reference effect IR supports file writes, messages, payments, database
mutations, deployments, and permission changes. Each clause binds operation,
exact target, recipient, data class, a nonnegative amount ceiling, payload
digest, postcondition contract, purpose, exact task atom, effect-contract
version, durability, and idempotency scope. The fixtures use synthetic
identifiers and payload digests only.

\subsection{Profiles and mutants}

Table~\ref{tab:profiles} defines mechanism families rather than reproductions
of cited systems. The labels describe which checks are intentionally present;
they must not be read as empirical results for SIGIL, SkillGuard, VIGIL,
CommitGuard, or another implementation.

\begin{table}[t]
\caption{Transparent reference profiles. A check mark means that the finite
reference evaluator performs the dimension.}
\label{tab:profiles}
\small
\begin{tabularx}{\linewidth}{@{}L{0.24\linewidth}cccccY@{}}
\toprule
Profile & Root & Closure & Effect & Commit refresh & Correlation & Purpose \\
\midrule
Metadata/manifest only & -- & -- & -- & -- & -- & Negative control \\
Root identity & \(\checkmark\) & -- & -- & -- & -- & Flat-root binding \\
Load-time closure & \(\checkmark\) & \(\checkmark\) & -- & -- & -- & Exact load snapshot \\
Effect-only commit & -- & -- & \(\checkmark\) & \(\checkmark\) & \(\checkmark\) & Effect boundary \\
Prepared closure + effect & \(\checkmark\) & \(\checkmark\) & \(\checkmark\) & -- & \(\checkmark\) & Composed prepare binding without commit-time closure/effect/state refresh \\
\ClosureBound & \(\checkmark\) & \(\checkmark\) & \(\checkmark\) & \(\checkmark\) & \(\checkmark\) & Full joint predicate \\
\bottomrule
\end{tabularx}
\end{table}

In Table~\ref{tab:profiles}, ``commit refresh'' marks a reduced profile that
performs any operation-time currentness checks, not necessarily the complete
frontier. The effect-only profile rechecks the current normalized effect against
the grant, exact epochs, task/time status, and reservation freshness, but omits
closure, authoritative pre-state, effect-stability, evidence-profile, and
canonicalizer checks. Only \ClosureBound performs the complete joint refresh.

Six seeded semantic mutants each weaken one full-profile mechanism: project
away one material transitive package, omit commit-epoch validation, omit
commit-root validation, bypass normalization for one trusted surface name,
replace relational clauses with field unions, or omit reservation freshness.
For every mutant, the unmodified full evaluator first rejects the same trace;
the mutant is counted as detected only if its single weakening then reaches the
targeted unsafe commit.

\subsection{Complexity}

For \(n=|V|\), \(m=|E|\), and total evidence bytes \(L_{\mathrm{evid}}\),
resolution and SCC construction are \(O(n+m+L_{\mathrm{evid}})\) plus adapter
lookup costs. Sorting commitments is
\(O(n\log n+m\log m)\). With immutable-node caching, re-resolution hashes only
changed evidence but still validates the authenticated currentness frontier.
Policy checking is linear in the number of non-dominated correlated clauses in
the simple reference representation.

Worst-case graph and clause growth are real. A production implementation should
use canonical antichains, exact dominance elimination, immutable subgraph
roots, and fail-closed profile limits. It must not handle exhaustion by
silently dropping dependencies or unioning restrictive clauses.

\section{Evaluation}
\label{sec:evaluation}

\subsection{Questions and evidence discipline}

The evaluation asks five bounded questions:

\begin{itemize}
  \item \textbf{RQ1: Conformance.} Does the reference profile make the expected
  decision on every frozen represented lifecycle case?
  \item \textbf{RQ2: Necessity.} Do reduced profiles and seeded semantic mutants
  expose the counterexamples predicted by Proposition~\ref{prop:necessity}?
  \item \textbf{RQ3: State safety.} Does exhaustive exploration of a declared
  finite lifecycle find an invariant violation in the full profile, and does
  it find witnesses for mutants?
  \item \textbf{RQ4: Local cost.} What is the narrow in-process cost of
  canonical commitment and the joint gate on one disclosed host?
  \item \textbf{RQ5: External structure.} Under a frozen, explicitly lexical
  audit, what local and cross-root path structure appears in two public Skill
  repository snapshots?
\end{itemize}

Questions about model adherence, hidden-dependency prevalence, real-runtime
sink coverage, and task utility require complementary empirical methods. The
external audit supplies repository-structure evidence and a baseline for
adapter-based runtime studies.

\subsection{Fixture construction}

The frozen suite has 40 synthetic cases: 32 unsafe rows across eight
four-case families and eight benign controls. Every row fixes:

\[
\langle G_{\mathrm{grant}},G_{\mathrm{prepare}},G_{\mathrm{commit}},
\gamma,x_{\mathrm{prepare}},x_{\mathrm{commit}},S,\textsf{expected}\rangle .
\]

The unsafe families are:

\begin{enumerate}
  \item root Skill or root-instruction mutation;
  \item nested Skill or package mutation with unchanged root;
  \item service-deployment or effect-contract epoch drift;
  \item explicitly inexact resolution, a floating version, unauthenticated
  node evidence, or a declared node made unreachable;
  \item prepare/commit closure or current-state races;
  \item amount, recipient, data-class, or effect-contract-version amplification;
  \item out-of-ceiling equivalent-effect paths through helper, shell, or browser
  surfaces; and
  \item a consumed reservation, revoked task, expired grant, or
  correlated-clause splicing.
\end{enumerate}

The eight benign controls cover canonical discovery order, canonical path
spelling, alpha-renamed resolver handles, an exact-ceiling payment, an
authorized helper path, a dependency update accompanied by a new grant, a fresh
service epoch, and one non-durable read-only effect. A control is counted
correct only if it reaches its expected allow/reauthorize/deny state; ``safe
because everything was denied'' would fail the suite.

\subsection{Kernel contract checks}

All 18 deterministic kernel contracts pass. They cover order and path
canonicalization, material-node changes, cross-SCC endpoint rewiring, root
changes inside an SCC, unreachable-node rejection, authenticated grant and
task/purpose binding, correlated-clause non-splicing, equivalent-path
normalization, unsupported-surface denial, cross-SCC edge-evidence changes, and
evidence-profile/canonicalizer downgrade. They also check that local-handle
alpha-renaming leaves the root unchanged and duplicate canonical-node handles
are interned, expected epoch maps require exact key/value equality, payload and
pre-state changes alter effect identity, and authenticated resolution context
alters node identity. K04 and K05 are direct regression
tests for the two collision-free equal-root constructions found during the
independent formal audit; K12 verifies that cross-component edge evidence, not
only its endpoints, enters \(\rho\); K14 and K15 guard representation
invariance; K16--K18 guard exact currentness and effect/context identity.

\subsection{Conformance and reduced profiles}

Table~\ref{tab:conformance} is generated from the aggregate JSON. The full
profile matched all 40 expected outcomes: all 32 unsafe rows were denied or
sent to reauthorization before durability, and all eight benign controls
received their exact expected decision. It had zero unsafe allow, false denial,
or decision mismatch on this finite suite.

\begin{table}[t]
\caption{Finite fixture results. Unsafe allow is the security error;
false denial counts benign rows denied contrary to their exact oracle.
Profiles are mechanism classes, not reproduced cited systems.}
\label{tab:conformance}
\scriptsize
\setlength{\tabcolsep}{3pt}
\begin{tabularx}{\linewidth}{@{}L{0.25\linewidth}rrrrY@{}}
\toprule
Profile & Exact & Unsafe allow & False deny & Decision err. & Main omitted protection \\
\midrule
Metadata/manifest only & 8 & 32 & 0 & 32 & Trusted recursive identity and effect gate \\
Root identity & 14 & 26 & 0 & 26 & Transitive closure and normalized effect \\
Load-time closure & 20 & 20 & 0 & 20 & Operation-time refresh and effect binding \\
Effect-only commit & 24 & 16 & 0 & 16 & Artifact/closure identity \\
Prepared closure + effect & 32 & 8 & 0 & 8 & Commit-time freshness/replay fencing \\
\ClosureBound & 40 & 0 & 0 & 0 & None within represented profile \\
\bottomrule
\end{tabularx}
\end{table}

These counts should not be ranked as detector accuracy. Cases are authored
litmus tests with exact structural oracles, and profile differences follow from
deliberately omitted checks.

\subsection{Mutation evidence}

All six seeded mutants produced a targeted unsafe witness:

\begin{table}[t]
\caption{Single-point full-profile mutants and minimized witness shape.}
\label{tab:mutants}
\small
\begin{tabularx}{\linewidth}{@{}L{0.28\linewidth}L{0.26\linewidth}Y@{}}
\toprule
Mutant & First violated invariant & Minimal witness \\
\midrule
Project out material package & I2 exact closure & Full resolver reauthorizes; projected resolver/grant omits changed package and commits \\
Omit commit epoch & I4 current operation & Full gate reauthorizes after redeploy; mutant accepts prepared epoch \\
Omit commit root & I2 exact closure & Full gate reauthorizes after helper update; mutant accepts prepared root \\
Trust shell surface & I5 path invariance & Full clause gate rejects; named shell path bypasses normalized policy \\
Union clause fields & I3 non-amplified effect & Full relation rejects; mutant combines recipient and amount from distinct clauses \\
Omit reservation freshness & I4 current operation & Full gate denies consumed nonce; mutant commits its replay \\
\bottomrule
\end{tabularx}
\end{table}

The \(6/6\) mutation score is targeted mechanism evidence for the six declared
faults, rather than general coverage of implementation errors.

\subsection{Bounded lifecycle exploration}

The explicit-state model contains phases
\[
\{\mathsf{idle},\mathsf{resolved},\mathsf{granted},
\mathsf{prepared},\mathsf{committed},\mathsf{denied},
\mathsf{replayDenied}\},
\]
two closure versions, two service epochs, two effect ceilings, two
purpose/recipient correlations, two surface paths, one revocation transition,
authenticated versus forged grants, and a single-consumption reservation.
From the initial state it enumerates every enabled update, resolve, grant,
prepare, revoke, commit, denial/restart, and replay-attempt transition until the
BFS queue is empty; no depth cutoff is used. Every profile permits the
environment to attempt replay: the full monitor reaches
\(\mathsf{replayDenied}\), while the reservation mutant accepts the second use
and records an I4 witness.

The full profile reached a fixed point after 84,608 unique states and 530,752
transitions; the maximum shortest-path depth was 23 and the remaining frontier
was zero. It found no violation of I1--I5, including untrusted-grant and
equivalent-surface states. Each of the six mutant profiles reached a violation,
and the artifact records its first breadth-first witness. These finite profiles
abstract the same six mechanism families as the executable mutants. Exploration
is exhaustive within the encoded finite abstraction; package-manager and
remote-service behavior require separate models. The manuscript theorems,
rather than the BFS count, support the general conditional claims.

\subsection{Local mechanism cost}

The microbenchmark runs after warm-up in one Node.js process on an Apple M4 Max
host. Inputs are preconstructed in memory; network access, package resolution,
filesystem reads, signatures, remote attestations, model inference, durable
transactions, and receiver settlement are excluded.

For 100,000 iterations over the fixed small fixture, the frozen reported run's
median in-process closure commitment is 0.045\,ms (p95 0.055\,ms), pure effect
normalization plus its correlated-clause gate is 0.002\,ms (p95 0.003\,ms),
and a full-profile fixture evaluation is 0.120\,ms (p95 0.146\,ms). The last
operation includes two closure commitments and two normalizations; the three
measurements are not additive pipeline stages. They establish that the
reference kernel is mechanically lightweight on the disclosed host. The result
characterizes in-process mechanism cost, excluding authoritative resolution,
distributed fencing, and other production operations.

\subsection{External repository-snapshot closure audit}

To add evidence independent of the authored lifecycle fixtures, we froze two
public distributions: 19 top-level Skills under \mcode{skills/*/SKILL.md} at
Anthropic repository commit
\mcode{3b3fad96} and 530 non-fixture
\mcode{plugins/**/skills/**/SKILL.md} files at OpenAI repository commit
\mcode{1e285826}
~\cite{anthropic2026skillsrepo,openai2026pluginsrepo}. The two repository-layout
adapters recursively enumerate regular files and symlinks below each selected
Skill root without following symlinks. They separately classify unique Markdown
link targets relative to that root and count a bundled file as lexically named
only when its full root-relative path occurs verbatim in \mcode{SKILL.md}.

\begin{table}[t]
\caption{External static closure audit at two fixed public commits. ``All paths
named'' is evaluated only among roots with bundled files. It is a strict lexical
measure, not semantic dependency recall.}
\label{tab:external-corpus}
\small
\begin{tabularx}{\linewidth}{@{}L{0.20\linewidth}rrrrY@{}}
\toprule
Source & Skills & Unique files & Bundled roots & All paths named & Roots with cross-root links \\
\midrule
Anthropic Skills & 19 & 411 & 18 & 6 & 0 \\
OpenAI Plugins & 530 & 4,461 & 508 & 15 & 67 \\
\textbf{Total} & \textbf{549} & \textbf{4,872} & \textbf{526} & \textbf{21} & \textbf{67} \\
\bottomrule
\end{tabularx}
\end{table}

The 549 roots contain 5,220 root--file instances but 4,872 unique local files
because some selected Skill roots nest other Skill roots. Among the 526 roots
with at least one non-\mcode{SKILL.md} file, only 21 name every such relative
path verbatim. Across 1,341 unique Markdown targets, 208 links in 67 Skills
resolve inside the same repository but outside the immediate Skill root; 17 are
absolute-path targets and 17 do not resolve in the frozen checkout. These last
categories are not automatically defects: a path may be documentation-rooted,
generated at use time, or intentionally external. None of the 549 frontmatters
contains a \mcode{dependencies} field, which is consistent with the current
format rather than a conformance failure.

The result supports a structural premise: hashing \mcode{SKILL.md} or assuming
exhaustive path mention is insufficient to reconstruct exact operational
closure in general. Lexical mention is neither necessary nor sufficient for
runtime loading, and repository enumeration excludes packages, services,
tools, generated paths, and runtime effects. The snapshots therefore add
externally authored structural evidence to the mechanism study and motivate
runtime validation of interoperability, prevalence, and discovery completeness.

\subsection{Validity and evaluation traceability}

\paragraph{Internal validity.}
Fixtures and implementation were authored together. Exact per-row oracles,
transparent profiles, semantic mutants, structural traces, and an artifact
verifier strengthen internal validity by exposing decision and binding errors.
The external corpus supplies independently authored inputs; its extractor and
lexical measures remain part of our study design.

\paragraph{External validity.}
The supported inference from the two frozen repositories is public distribution
structure at two commits. Production clients, package ecosystems, MCP
deployments, models, marketplaces, and runtime behavior require separate
empirical sampling; dependency completeness and sink coverage remain
deployment-specific.

\paragraph{Construct validity.}
``Unsafe allow'' means a represented durable effect violated the frozen joint
predicate. It is not a judgment that a Skill is malicious. ``False denial''
is interpreted against the authored benign oracle and does not estimate broad
user utility.

\paragraph{Evaluation traceability.}
The internal evaluation records Node/OS/CPU metadata, fixture and source
digests, profile counts, per-case traces, bounded counterexamples, and
microbenchmark parameters. The reviewed snapshot was rerun against frozen
inputs, and an internal verifier checked both result sets and cross-checked
every manuscript count. These controls establish snapshot identity and internal
consistency.

\section{Closest-Work Comparison}
\label{sec:comparison}

Table~\ref{tab:closest} compares guaranteed axes, not headline system names.
A check mark means the cited work directly supplies the property according to
its published description; a circle means partial or adjacent support.

\begin{table}[t]
\caption{Closest-work capability matrix. ``Exact heterogeneous closure''
requires authenticated operational Skill/package/tool/model/service structure;
``joint witness'' binds that closure and realized effect in one durability-time
authorization record.}
\label{tab:closest}
\scriptsize
\setlength{\tabcolsep}{2.4pt}
\begin{tabularx}{\linewidth}{@{}L{0.20\linewidth}cccccY@{}}
\toprule
Work & \shortstack{Exact hetero.\\artifact closure} &
\shortstack{Runtime\\authority} &
\shortstack{Durability\\refresh} &
\shortstack{Realized-effect\\normalization} &
\shortstack{Joint\\witness} & Primary scope \\
\midrule
Skills Are Not Islands~\cite{jia2026islands} & \(\circ\) & -- & -- & -- & -- & Measurement and dependency management \\
SkillFortify~\cite{bhardwaj2026skillfortify} & \(\circ\) & \(\circ\) & -- & -- & -- & Static analysis, SAT resolution, lockfile \\
OpenAPM v0.1~\cite{microsoft2026openapm} & \(\circ\) & -- & -- & -- & -- & Install-time package/lock/policy draft \\
SIGIL~\cite{shen2026sigil} & \(\circ\) & \(\circ\) & -- & -- & -- & Audit-to-load integrity \\
SkillGuard~\cite{pan2026skillguard} & \(\circ\) & \(\checkmark\) & \(\circ\) & \(\circ\) & -- & Skill permission framework \\
Edge Skillguard~\cite{zhan2026autopolicy} & \(\circ\) & \(\checkmark\) & \(\circ\) & \(\circ\) & -- & Physical-world state/sensor guards \\
VIGIL~\cite{li2026vigil} & \(\circ\) & \(\checkmark\) & \(\circ\) & \(\circ\) & -- & Trace behavioral enforcement \\
SkillScope~\cite{wu2026skillscope} & \(\circ\) & \(\checkmark\) & \(\circ\) & \(\circ\) & -- & Task-conditioned action graph \\
Dynamic Capabilities~\cite{zhou2026dynamic} & \(\circ\) & \(\circ\) & \(\circ\) & -- & -- & Manifest certificate and ledger \\
AIRGuard~\cite{qin2026airguard} & -- & \(\checkmark\) & \(\circ\) & \(\circ\) & -- & Pre-action normalized authority control \\
Cordon~\cite{chen2026cordon} & \(\circ\) & \(\checkmark\) & \(\checkmark\) & \(\circ\) & \(\circ\) & Transaction lineage/state/effect validation \\
CommitGuard~\cite{santos2026temporary} & -- & \(\checkmark\) & \(\checkmark\) & \(\circ\) & \(\circ\) & General durable-effect authority \\
\ClosureBound & \(\checkmark\) & \(\checkmark\) & \(\checkmark\) & \(\checkmark\) & \(\checkmark\) & Exact closure/effect non-inheritance \\
\bottomrule
\end{tabularx}
\end{table}

\paragraph{Direct complementarity and novelty boundary.}
SkillFortify and OpenAPM already supply substantial dependency-resolution,
lockfile, typed-package, and install-policy machinery; AIRGuard already
normalizes heterogeneous tool actions for current authority; Cordon already
validates lineage, authority, staged state, and pending effects as a composed
transaction; and CommitGuard already establishes durability-time authority
refresh. Edge Skillguard already adds typed world-state and sensor-evidence
guards to physical-world Skills. ClosureBound's novelty lies in the joint
construction: to the best of our
knowledge, prior work has not jointly bound an authenticated heterogeneous
operational-artifact closure and a cross-surface normalized \emph{realized}
effect in one durability-time authorization witness, nor established both
closure-version non-inheritance and equivalent-effect path invariance under
explicit completeness assumptions.

\section{Related Work}
\label{sec:related}

\subsection{Agent Skill security and measurement}

Large-scale analysis of 31,132 Skills found vulnerabilities across prompt
injection, exfiltration, privilege escalation, and supply-chain categories
~\cite{liu2026wild}. Lifecycle taxonomies emphasize missing data/instruction
boundaries, persistent trust, and marketplace review
~\cite{li2026towards}; SkillSec-Eval extends lifecycle evaluation through
repository, retrieval, planning, execution, and evolution over 327 reported
real-world Skills~\cite{badhe2026skillsec}. MalSkillBench contributes
runtime-verified malicious and
matched benign Skills, while AgentTrap measures complete trajectories and
routine-looking unsafe side effects
~\cite{guo2026malskill,zhuang2026agenttrap}. Runtime Skill Audit dynamically
probes risk-relevant contexts~\cite{lan2026runtime}. SkillDetonate lifts
instructions observed during sandboxed execution and tracks OS-boundary
information flow; its ``closure'' is runtime instruction/flow coverage, not a
version commitment~\cite{ji2026detonate}. Composition studies further show that
benign-looking Skills can become harmful along activated paths
~\cite{xie2026composition}. These works establish the threat and evaluation
ecosystem. ClosureBound is not a detector and does not reuse their rates as its
evidence.

Natural-language surfaces create distinct risks. Under the Hood of
\mcode{SKILL.md} studies discovery, selection, and governance manipulation;
payload-less Skills induce harmful generated behavior without recognizable
payloads; BadSkill hides backdoors in model artifacts
~\cite{saha2026hood,liu2026payloadless,tie2026badskill}. These results motivate
metadata non-authority, model nodes, and effect mediation. They do not imply
that exact identity makes harmful content safe.

\subsection{Skill permissions, least privilege, and behavioral enforcement}

SkillGuard connects manifests, context influence, action permissions, user
authorization, and runtime monitoring~\cite{pan2026skillguard}. SkillScope
extracts task-conditioned action graphs and constrains over-privileged paths
~\cite{wu2026skillscope}. VIGIL compiles temporal, argument, and value-flow
specifications to finite-trace enforcement~\cite{li2026vigil}. A distinct
SIGIL system compiles natural-language Skills into typed executable harnesses
and evaluates mandate compliance across public Skills and three runtime models
~\cite{dantanarayana2026sigilcompiler}. Edge Skillguard co-packages typed
invocation policy over world state and sensor evidence for physical actuation
~\cite{zhan2026autopolicy}. These systems
govern what behavior is permitted. ClosureBound addresses whether the exact
recursive implementation and effect contract being used now are the ones for
which that permission was granted.

SIGIL binds audited Skill artifacts to a mandatory verified loader
~\cite{shen2026sigil}; dynamic-capability certificates bind a manifest hash to
agent identity~\cite{zhou2026dynamic}. ClosureBound builds on both
artifact-binding contributions. Its later enforcement point and typed
transitive/effect joint witness address authority transfer after loading.

\subsection{Software supply-chain integrity}

TUF secures software-update metadata under key compromise; in-toto verifies
intended supply-chain steps and materials; SLSA specifies increasingly strong
provenance and build guarantees
~\cite{cappos2010tuf,torresarias2019intoto,slsa2026}. SPDX and CycloneDX
represent components, services, and dependency relationships
~\cite{spdx2025,cyclonedx2025}. ClosureBound uses these as evidence formats and
locates its contribution in operation-time authorization rather than SBOM,
provenance, lockfile, attestation, or Merkle construction.
SkillFortify formalizes an Agent Dependency Graph, capability-aware SAT
resolution, deterministic Skill lockfiles, and ASBOMs
~\cite{bhardwaj2026skillfortify}; its August 2026 v2 also corrects bibliography,
external claims, and the empirical boundary of its information-flow analysis.
The OpenAPM v0.1 working draft specifies
transitive package/MCP resolution, typed primitives, lockfiles, tree hashes, and
install-time governance while explicitly leaving runtime harness APIs out of
scope~\cite{microsoft2026openapm}. ClosureBound begins after these
contributions: it asks whether a trusted task grant remains valid for the
authenticated operational closure and realized effect at durability.

\subsection{Runtime and usage control}

Reference-monitor complete mediation and least privilege are classical design
principles~\cite{saltzer1975protection}. Security automata characterize
execution monitoring~\cite{schneider2000enforceable}; Polymer composes runtime
policies~\cite{bauer2005polymer}; UCON models ongoing authorization,
obligations, conditions, and mutable attributes~\cite{park2004ucon}. Cedar
provides expressive, analyzable authorization with a verified core
~\cite{cedar2024}. ClosureBound is a specialized policy/state model built on
these foundations.

FORGE provides recursive Datalog policies, causal-history predicates, an
assume/guarantee observability service, and a reference monitor
~\cite{palumbo2026forge}. AIRGuard provides pre-action runtime authority control
over normalized heterogeneous tool actions~\cite{qin2026airguard}. Atomix uses
epochs and resource frontiers for transactional release of tool effects, while
Cordon validates result lineage, delegated authority, staged state, and pending
effects before commit~\cite{mohammadi2026atomix,chen2026cordon}. Temporary
Authority, Permanent Effects most directly establishes freshness, causal
priority, effect binding, and eligibility at the durability boundary
~\cite{santos2026temporary}. ClosureBound's contribution is the
Agent-Skill-specific conjunction of an exact
recursive heterogeneous artifact identity and cross-surface realized-effect
equivalence bound to that current durability decision.

\subsection{Relationship to public prior team work}

Three publicly accessible team papers provide adjacent design context.
OpenPort supplies protocol and governed-effect vocabulary; IGAC supplies an
independently derived current-task authorization certificate; and EBTE
separates model-authored explanations from server-verified action claims
~\cite{zhu2026openport,zhu2026igac,zhu2026ebte}. Table~\ref{tab:internal}
makes the reuse boundary explicit.

\begin{table}[t]
\caption{Public team-work reuse boundary. Prior observations and denominators
are not ClosureBound evidence.}
\label{tab:internal}
\scriptsize
\begin{tabularx}{\linewidth}{@{}L{0.19\linewidth}L{0.32\linewidth}Y@{}}
\toprule
Work & Adjacent public abstraction & Not re-claimed here \\
\midrule
OpenPort & Protocol facts, policy/effect vocabulary, conformance discipline & Protocol novelty or published runtime evidence \\
IGAC & Authenticated current-task certificate & Intent inference and IGAC empirical results \\
EBTE & Server-owned claims and digest-bound evidence & Explanation verification novelty and observations \\
\bottomrule
\end{tabularx}
\end{table}

ClosureBound is a distinct identity/effect unit over this public context. No
non-public manuscript contributes a premise, artifact row, model run, bounded
state, or empirical denominator. The 40-case, mutation, bounded-model,
performance, and external-corpus results reported in
Section~\ref{sec:evaluation} are generated solely by this paper's reported
evaluation configuration and internally verified receipts.

\section{Discussion}
\label{sec:discussion}

\subsection{Security value and residual risk}

ClosureBound changes the default unit of trust from ``Skill \(s\) was once
approved'' to ``this exact current recursive capability is authorized to
produce this exact class of effect now.'' It is valuable against silent
dependency substitution, service/schema drift, load-to-use races, path
laundering, and stale authorization replay. Content benignness, provider honesty
under weak evidence, and unmediated channels remain separate security controls.

\subsection{Completeness is a deployment obligation}

Assumption~\ref{ass:closure} is intentionally demanding. Agent Skills allow
arbitrary prose and scripts, and measurement shows heterogeneous hidden
dependencies~\cite{jia2026islands}. A deployment should expose a completeness
frontier rather than a Boolean marketing claim:

\[
\mathsf{coverage}=
\langle nodeKinds,sinks,dynamicLoads,evidenceClasses,unknowns\rangle .
\]

High-risk policy may require ``no unknown material dynamic loads'' and
attested service deployments. Lower-risk tasks may accept a pinned container
as a conservative enclosing identity. An unversioned service URL is not silently
upgraded to strong evidence.

\subsection{Availability and update ergonomics}

Exact non-inheritance means updates cause reauthorization. This is a security
property and an availability cost. Typed contract/evidence diffs, reusable
subgraph roots, risk-tiered evidence profiles, and policy-authorized version
ranges can reduce friction. A range is itself part of the grant and must be
evaluated against a trusted compatibility/effect contract; \mcode{@latest} is
not an exact range.

Automatic renewal is safe only when a trusted policy already authorizes the
new denotation and the new grant binds the new root. Calling it ``same grant''
would hide the authority transition.

\subsection{Remote services and settlement}

Three service identities should not be confused:

\begin{enumerate}
  \item endpoint/audience identity;
  \item declared schema/effect-contract identity; and
  \item deployed implementation identity.
\end{enumerate}

Many services expose only the first two. ClosureBound can prevent token
audience confusion and schema drift, consistent with MCP authorization
guidance~\cite{mcp2026auth}. Deployed implementation identity requires
independent evidence. Under MCP
2026-07-28, stateless per-request protocol context, server discovery, issuer
binding, and cache TTL/scope become explicit evidence inputs, not substitutes
for deployment identity or operation-time authorization
~\cite{mcp2026changelog}. Confidential or
irreversible profiles should require deployment attestation, a constrained
server-side transaction contract, or human confirmation.

An authenticated receiver receipt establishes only the receiver-reported
state unless an independent postcondition or receiver-soundness premise is
available. A lost response after possible durability is indeterminate, not
safe to repeat automatically.

\subsection{Effect adapters as a governance API}

The adapter registry is a major TCB surface. Its contracts should be versioned,
reviewed, fuzzed, and included as closure nodes. A general command adapter can
often predict filesystem/network effects through a sandbox plan, but commands
with opaque dynamic behavior may require a coarse ceiling such as ``may write
any file under directory \(d\) and contact domain set \(R\).'' Coarse
normalization preserves soundness at the cost of false denial.

\subsection{Policy and legal interpretation}

Purpose, recipient, and data-class labels are policy inputs. ClosureBound's
theorems cover conditional enforcement; latent intent, legal validity, and
social desirability require separate human and institutional judgments.
Ambiguity should produce clarification or review rather than be resolved by
Skill prose.

\subsection{Ethics and dual use}

The evaluation uses synthetic identifiers, hashes, and state transitions; the
preprint does not publish live malicious Skills, credentials, or exploit payloads.
Closure and effect traces can reveal sensitive architecture, so production
telemetry should minimize payloads, separate authorization evidence from
content, apply retention limits, and enforce tenant isolation.

\section{Limitations and Future Work}
\label{sec:limitations}

The current work has six principal limitations.

\begin{enumerate}
  \item \textbf{Conditional completeness.} The proofs assume a sound
  dependency frontier and protected-sink coverage. Natural-language and
  reflective dynamic dependencies require deployment-specific instrumentation
  and conservative fallback.
  \item \textbf{Finite effect semantics.} Path invariance covers registered
  adapters, not arbitrary code equivalence or covert channels.
  \item \textbf{Service evidence.} A stable endpoint/contract cannot establish
  unchanged remote behavior.
  \item \textbf{Finite and static evidence.} Authored fixtures and bounded
  exploration establish mechanism behavior under the evaluated profiles. The
  two-snapshot external corpus contributes independently authored inputs and
  lexical structural evidence; cross-runtime efficacy requires execution-based
  adapter evaluation.
  \item \textbf{Single-process cost.} The reported microbenchmark excludes
  network, storage, signatures, attestations, distributed ordering, and model
  latency.
  \item \textbf{No receiver theorem.} The monitor authorizes dispatch and
  conditionally interprets settlement; it does not guarantee exactly-once
  external effects or compensation.
\end{enumerate}

The highest-value next studies are two independent runtime adapters, a broader
probability-sampled Skill corpus with runtime-observed exact-resolution
coverage, a matched
translation of AgentTrap or MalSkillBench to effect contracts, a distributed
fencing prototype, and an independently encoded TLA+/Alloy/Lean model. Together,
these studies would extend the present mechanism evidence into cross-runtime,
ecosystem, distributed-systems, and independently encoded formal validation.

\section{Conclusion}
\label{sec:conclusion}

Agent Skills turn procedural text and software into capabilities that are
resolved progressively and may change transitively. Authorizing the root name,
directory, manifest, or load-time artifact is therefore not enough to
authorize a later durable effect. ClosureBound binds a trusted task grant to
the exact current heterogeneous dependency closure and the normalized real
effect in one operation-time witness. Under explicit completeness, mediation,
freshness, cryptographic, and linearization assumptions, material updates do
not inherit authority, represented effects do not amplify transitive
privilege, distinguishable bound-value mismatches at commit fail closed, and
supported equivalent paths are checked against the same effective policy
denotation and cannot enlarge authority by path choice. The design
complements---rather than
relabels---Skill scanning, permissions, verified loading, behavioral
monitoring, SBOM/provenance, and general commit-time authorization. The external
snapshot audit provides structural motivation for conservative local closure in
current public Skill distributions and a concrete baseline for validating
runtime completeness across adapters and deployments.

\appendix

\section{Notation and Trust Summary}
\label{app:notation}

\begin{table}[ht]
\caption{Core notation.}
\small
\begin{tabularx}{\linewidth}{@{}L{0.2\linewidth}Y@{}}
\toprule
Symbol & Meaning \\
\midrule
\(\Graph_t(s;\pi,c_\pi,B)\) & Typed transitive capability graph under profile/canonicalizer \((\pi,c_\pi)\) and candidate effects \(B\) \\
\(q_t(v)\) & Canonical authenticated evidence record for node \(v\) \\
\(\rho_t(s;\pi,c_\pi,B)\) & Root-, context-, and endpoint-sensitive SCC/Merkle graph commitment \\
\(\gamma\) & Trusted correlated task/closure/effect grant \\
\(\Norm(x,S)\) & Partial trusted surface-to-effect normalizer \\
\(A_{\mathrm{cur}}(t,e)\) & Current task, closure, state, and purpose/provenance meet \\
\(W_{t_p}\) & Immutable prepared joint witness and reservation binding \\
\(\equiv_\Omega\) & External-effect equivalence under observation map \(\Omega\) \\
\bottomrule
\end{tabularx}
\end{table}

\begin{table}[ht]
\caption{Trusted and untrusted inputs.}
\small
\begin{tabularx}{\linewidth}{@{}L{0.28\linewidth}YY@{}}
\toprule
Input & Trust treatment & Failure behavior \\
\midrule
Skill name/description/prose & Untrusted discovery/instruction data & Cannot mint authority \\
Skill manifest/lockfile & Evidence only; authenticated if possible & Missing/unknown material edge denies exact profile \\
Supply-chain attestation & Authenticated identity/provenance evidence & Invalid/stale evidence denies \\
Task/purpose certificate & Trusted issuer input & Missing/ambiguous denies or asks for clarification \\
Tool/service schema & Versioned closure/effect-contract node & Drift requires new grant \\
Effect adapter & TCB, versioned and closure-bound & Unknown/ambiguous normalization denies \\
Receiver receipt & Authenticated receiver report unless stronger premise & Uncertain external state remains indeterminate \\
\bottomrule
\end{tabularx}
\end{table}

\section{Proof Detail}
\label{app:proofs}

\subsection{Correlated clauses}

Let a finite effect clause be
\[
\begin{aligned}
c=\langle&
taskAtom,purpose,recipient,op,target,dataClass,amountCeiling,\\
&argsDigest,postconditionContract,effectContractVersion,\\
&durability,idempotencyScope\rangle .
\end{aligned}
\]
Its denotation is the set of normalized effects satisfying every field
together. The effect's authoritative \(\mcode{preStateDigest}\) is bound
separately by (A5), while grant validity and epoch maps are bound by (A1)--(A3).
The formal clause language permits exact values, finite sets, prefix patterns,
and closed numeric intervals with an explicit bottom; this language is closed
under exact field intersection. The executable kernel instantiates exact-valued
fields and a nonnegative amount ceiling over already compiled correlated
clauses; it does not implement general policy intersection or dominance
elimination. A deployment using a richer language must retain the exact
symbolic conjunction or fail closed whenever an intersection is not
representable. For two policies \(C_1,C_2\), define
\[
C_1\otimes C_2=
\operatorname{Min}\{c_1\cap c_2:
c_1\in C_1,c_2\in C_2,\denote{c_1\cap c_2}\ne\emptyset\},
\]
where \(\operatorname{Min}\) removes exact duplicates and clauses whose
denotation is contained in another retained clause with identical authority
provenance.

\begin{lemma}[Relational meet soundness]
\label{lem:relmeet}
\[
\denote{C_1\otimes C_2}=\denote{C_1}\cap\denote{C_2}.
\]
\end{lemma}

\begin{proof}
If \(e\) is in the left side, it satisfies some nonempty pairwise intersection
and therefore both parent clauses. If \(e\) is in both parent denotations, it
satisfies some \(c_1\in C_1\) and \(c_2\in C_2\), hence their nonempty
intersection before dominance elimination. Removing a subset clause does not
remove \(e\) because a retained dominating clause contains it.
\end{proof}

Field-wise union lacks this result because it can construct tuples present in
neither parent relation. This is the clause-splicing counterexample in the
artifact.

\subsection{Induction over nested dependencies}

Fix a normalized effect \(e\) and its authenticated influence subgraph
\(\mathcal I_{\Graph}(e)\). Let \(P(X)\) state that every mediated realization
of \(e\) reachable through SCC \(X\) is bounded by \(\gamma.A\) and by the meet
of every applicable contract on all \(X\)-to-sink influence paths. SCC
condensation gives a DAG. At a sink component, the trusted adapter normalizes
the effect and (A4) checks its local applicable contracts, so \(P(X)\) holds.
Assume \(P\) for every successor component. The current component's exact
outgoing edges, original endpoints, successor components, and applicable
contract nodes are committed by \(\rho\); the compiler forms their relational
meet without field splicing, and the same commit predicate mediates each
dispatch. Thus the reachable effect set is a subset of the accumulated meet,
establishing \(P(X)\). Induction over the condensation DAG proves
Theorem~\ref{thm:nonamplification}; an internal cycle is treated as one SCC and
all of its applicable contracts enter the same meet.

\section{Evaluation Configuration and Availability}
\label{app:artifact}

The evaluated snapshot uses the following disclosed configuration:

\begin{description}[leftmargin=1.25em,style=nextline]
  \item[Mechanism evidence.]
  Forty represented lifecycle fixtures, 18 kernel contracts, six mechanism
  profiles (one full and five reduced), and six single-point semantic mutants.
  \item[Bounded exploration.]
  Breadth-first enumeration to a fixed point with no depth cutoff over the
  finite state variables specified in Section~\ref{sec:evaluation}; the full
  profile reaches 84,608 states and 530,752 transitions.
  \item[Local cost.]
  A single-process, in-memory measurement on the disclosed Apple M4 Max host,
  excluding network, resolution I/O, signatures, durable transactions, model
  inference, and settlement.
  \item[External structure.]
  A lexical audit of the Anthropic Skills snapshot at commit
  \mcode{3b3fad96} and the OpenAI Plugins snapshot at commit
  \mcode{1e285826}, using the inclusion and path-resolution rules stated in
  Section~\ref{sec:evaluation}.
\end{description}

The arXiv source bundle contains the manuscript and bibliography; the evaluated
implementation, private source layout, machine-readable traces, and internal
verification commands are retained internally. The manuscript discloses the
methods, denominators, aggregate results, and limitations, while independent
execution requires access to the reviewed implementation beyond that bundle.

\section{Claim Checklist}
\label{app:claims}

\begin{itemize}
  \item Every ``all,'' ``zero,'' and state count refers only to its declared
  finite fixture, mutation, or bounded domain.
  \item Mechanism profiles are not reported as measurements of cited systems.
  \item Preprints are cited as contemporaneous work.
  \item Exact identity does not imply benignness or behavioral equivalence.
  \item Commit authorization is not receiver truth or exactly-once settlement.
  \item Unresolved dependencies and effects fail closed; they are not omitted.
  \item The two-source static audit supports repository-structure claims;
  cross-runtime utility, ecosystem coverage, and production overhead require
  separate evaluation.
\end{itemize}

\bibliographystyle{ACM-Reference-Format}
\setlength{\bibsep}{-0.2pt}
\bibliography{references}

\end{document}